\def\JMPBUILD{0}

\ifnum\JMPBUILD=1
  \documentclass[aip,jmp,amsmath,amssymb,preprint,nofootinbib]{revtex4-2}
\else
  \documentclass[11pt,a4paper]{article}
\fi

\usepackage[T1]{fontenc}
\usepackage[utf8]{inputenc}
\usepackage{amsmath,amssymb,amsthm,mathtools,mathrsfs}
\usepackage{enumitem}
\usepackage{xcolor}
\usepackage{tikz}
\usetikzlibrary{arrows.meta,patterns}

\ifnum\JMPBUILD=0
  \usepackage{lmodern}
  \usepackage[margin=2.6cm]{geometry}
  \usepackage{caption}
  \usepackage{microtype}
  \usepackage[numbers,sort&compress]{natbib}
  \usepackage[raggedright]{titlesec}
  \usepackage{needspace}
  \usepackage{placeins}
\else
  \providecommand{\needspace}[1]{}
  \providecommand{\FloatBarrier}{\relax}
\fi

\usepackage[hidelinks]{hyperref}
\usepackage{orcidlink}
\usepackage[capitalize,nameinlink]{cleveref}

\ifnum\JMPBUILD=0
  \usepackage{lmodern}
  \usepackage[margin=2.6cm]{geometry}
  \usepackage{caption}
  \usepackage{microtype}
  \usepackage[numbers,sort&compress]{natbib}
  \usepackage[raggedright]{titlesec}
  \usepackage{needspace}
  \usepackage{placeins}
\else
  \providecommand{\needspace}[1]{}
  \providecommand{\FloatBarrier}{\relax}
\fi

\usepackage{sansmath}

\ifnum\JMPBUILD=0
  \DeclareCaptionFont{sanscaptionmath}{\sansmath}
\fi

\ifnum\JMPBUILD=1 \else
                  \fi

\allowdisplaybreaks
\providecommand{\doi}[1]{doi:\allowbreak\href{https://doi.org/#1}{\nolinkurl{#1}}}
\definecolor{figblue}{RGB}{0,102,204}
\definecolor{figred}{RGB}{204,0,0}
\definecolor{gapgreen}{RGB}{0,153,76}
\definecolor{figpurple}{RGB}{128,0,128}

\newcommand{\ed}{d}
\newcommand{\cl}{\operatorname{cl}}
\newcommand{\id}{\mathbf{1}}
\newcommand{\ket}[1]{\lvert{#1}\rangle}
\DeclareMathOperator{\Tr}{Tr}
\DeclareMathOperator{\intr}{int}

\newtheoremstyle{colonplain}
  {\topsep}{\topsep}{\itshape}{}{\bfseries}{:}{.5em}{}
\newtheoremstyle{colondefinition}
  {\topsep}{\topsep}{\normalfont}{}{\bfseries}{:}{.5em}{}
\newtheoremstyle{colonremark}
  {0.5\topsep}{0.5\topsep}{\normalfont}{}{\itshape}{:}{.5em}{}

\theoremstyle{colondefinition}
\newtheorem{definition}{Definition}[section]
\newtheorem{assumption}[definition]{Assumption}
\newtheorem{convention}[definition]{Convention}

\theoremstyle{colonplain}
\newtheorem{theorem}{Theorem}[section]
\newtheorem{lemma}[theorem]{Lemma}
\newtheorem{proposition}[theorem]{Proposition}
\newtheorem{corollary}[theorem]{Corollary}

\theoremstyle{colonremark}
\newtheorem{remark}[theorem]{Remark}

\makeatletter
\renewenvironment{proof}[1][\proofname]{\par
  \pushQED{\qed}%
  \normalfont \topsep6\p@\@plus6\p@\relax
  \trivlist
  \item[\hskip\labelsep
        \itshape
    #1:]\ignorespaces
}{%
  \popQED\endtrivlist\@endpefalse
}
\makeatother

\newcommand{\ABSTRACTTEXT}{%
We investigate the relation between fundamental complements of Bousso and Kaya and regional operator algebras in global de Sitter space.  For a finite union of open arcs on the time-symmetric circle of global \(\mathrm{dS}_2\), we show that the fundamental complement is determined by the largest complementary gap: a gap contributes precisely when its angular length is at least \(\pi\), and at most one gap can do so.  In higher-dimensional global de Sitter space, the corresponding criterion is containment of an open hemisphere.

We then study the operator-algebraic consequences in an abstract M\"obius-covariant chiral conformal net. Each regional algebra is crossed with its vacuum modular flow using a shared auxiliary clock.  We derive an exact criterion for inclusions between the resulting continuous cores and show that vacuum correlations obstruct this criterion when the larger region contains an additional spacelike arc.  A compatible family can nevertheless be obtained on the finite Boolean algebra generated by a fixed collection of separated arcs by using a split product state.

Motivated by the holograms prescription of Bousso and Penington, we introduce a representation-dependent commutant model for regional algebras.  When the protected gap overlaps a fixed reference arc, the assigned algebra is a proper von Neumann subalgebra of the reference type-\(\mathrm{II}_1\) factor.  For an explicit two-arc family, adding an arbitrarily short antipodal component removes the fundamental complement and changes the assigned algebra discontinuously to the full reference algebra.%
}

\newcommand{\KEYWORDTEXT}{de Sitter space, JT gravity, fundamental complement, hologram trace, operator algebras, algebraic quantum field theory, modular theory, generalized entanglement wedge}

\newcommand{\SHORTTITLE}{Toward an Observable Algebra for Gravitating de Sitter Space}

\ifnum\JMPBUILD=1
  \title{Toward an Observable Algebra for de Sitter Space: Gap Protection and Modular Dressings}
  \author{Hassan ElSayed}
  \email{hassan.elsayed@wisc.edu}
  \affiliation{\textit{Department of Physics, University of Wisconsin-Madison, Madison, WI 53706, U.S.A.}}
  \date{\today}
\else
  \title{\textbf{Toward an Observable Algebra for de Sitter Space: \\[2pt] Gap Protection and Modular Dressings \\[10pt]}}
  \author{Hassan ElSayed\,\orcidlink{0000-0003-1451-0186}\\[2pt]
    \small Department of Physics, University of Wisconsin--Madison,
    Madison, WI 53706, U.S.A.\\
    \small \texttt{hassan.elsayed@wisc.edu}}
  \date{\today}
\fi

\hypersetup{%
  pdftitle={\SHORTTITLE},
  pdfauthor={Hassan ElSayed},
  pdfkeywords={\KEYWORDTEXT}}

\begin{document}

\ifnum\JMPBUILD=1
  \begin{abstract}\ABSTRACTTEXT\end{abstract}
  \keywords{\KEYWORDTEXT}
  \maketitle
\else
  \maketitle
  \begin{abstract}\noindent\ABSTRACTTEXT\end{abstract}
  \medskip
  \noindent\textbf{Keywords:} \KEYWORDTEXT.
  \tableofcontents
\fi

\section{Introduction}
\label{sec:introduction}

Subregion duality gives one of the clearest links between geometry and quantum information in AdS/CFT \cite{Maldacena,GKP,Witten}. A boundary region is associated with a bulk entanglement wedge, whose semiclassical boundary is determined by the Ryu--Takayanagi prescription and its covariant and quantum extensions \cite{RyuTakayanagi,HRT,EngelhardtWall}. The corresponding operator statement is naturally formulated in terms of relative entropy and quantum error correction \cite{JLMS,DHW,ADH,HarlowTASI}. Extending this picture to a finite region in a closed universe is harder. There is no asymptotic boundary on which to anchor the construction, and gravitational constraints obstruct a naive tensor-product decomposition of the Hilbert space \cite{DonnellyFreidel}.

Bousso and Penington proposed a generalized entanglement wedge for gravitating regions \cite{BP}, following earlier work on the entropy of finite gravitating subsystems \cite{BalasubramanianCummings}. Bousso and Kaya refined this proposal by introducing the \emph{fundamental complement}, the part of the spacelike complement generated by timelike curves of infinite proper time \cite{BK}. This refinement isolates degrees of freedom that must remain complementary to the original region. Kaya, Rath, and Ritchie subsequently derive a closely related hollowing prescription for generalized entanglement wedges directly from the gravitational path integral \cite{HollowGrams}. The refined Bousso-Penington prescription has also been applied to a slowly evaporating black hole in JT-de Sitter space \cite{Y2H-Page_from_holograms}.  In that construction, the radiation remains gravitating, and the prescription is applied slice by slice in the adiabatic regime. This yields, to our knowledge, the first explicit Page curve derived from the Bousso-Penington proposal.

A second line of work has been operator-algebraic. Leutheusser and Liu showed how a type-\(\mathrm{III}_1\) von Neumann algebra, and with it an emergent bulk time, arises in the large-\(N\) limit of holographic duality \cite{LeutheusserLiu2021,LeutheusserLiu2022}; Witten observed that imposing the gravitational constraint replaces such an algebra by a crossed product of type \(\mathrm{II}\) \cite{Witten2022}; and Chandrasekaran, Longo, Penington, and Witten showed that adjoining an observer clock to a static-patch algebra does the same in de Sitter space \cite{CLPW}. Penington and Witten then carried the programme through for JT gravity itself: in the setting with an asymptotic boundary, canonical quantization gives an algebra of boundary observables that is commutative in pure JT gravity and becomes a type-\(\mathrm{II}_\infty\) factor with trivial centre once bulk matter is included, on which density matrices and entropies are determined up to a rescaling and a shift respectively \cite{PeningtonWitten2023}. Together these developments suggest that the causal notion of a fundamental complement should have an operator-algebraic counterpart.

This paper investigates that relation in global \(\mathrm{dS}_2\), whose time-symmetric spatial slice is a circle. We combine a geometric classification of fundamental complements with a representation-dependent operator-algebraic model of regional observables. The latter is formulated for a chiral conformal net and isolates the compatibility conditions that a gravitational realization would have to satisfy.

\subsection{Results}
\label{sec:intro-results}

Let \(a\subset S^1\) be an admissible spatial region in the sense of \Cref{def:admissible-region} --- a finite union of open arcs with pairwise disjoint closures --- let \(A:=D(a)\) be its domain of dependence, and let \(b_1,\dots,b_m\) be the connected components of the regular-open complement \(a^\perp:=\intr(S^1\setminus a)\), with angular lengths \(\beta_k\). Our first result (\Cref{thm:gap-protection}) is the gap-protection rule
\begin{align}
  \widetilde A
  &=\bigcup_{\beta_k\geq\pi}D(b_k),
  \quad
  \widetilde a:=\widetilde A\cap\Sigma_0 ,
  \label{eq:intro-gap-rule}
\end{align}
where \(\widetilde A\) is the spacetime fundamental complement of \(A\) and \(\widetilde a\) its spatial trace. A complementary gap therefore contributes to the fundamental complement exactly when it occupies at least half of the circle. At most one gap can do so, which makes the multi-interval problem unusually simple: the answer depends on the largest gap, not on the total measure of \(a\). In particular a short arc placed near the antipode of another short arc splits the remaining complement into two unprotected gaps. Because the union in \eqref{eq:intro-gap-rule} has at most one term, it also satisfies \(\widetilde A=D(\widetilde a)\). \Cref{sec:higher-dimensions} shows that the threshold itself is dimension independent. For an open \(g\subseteq S^{d-1}\) the domain of dependence \(D(g)\) carries a complete timelike curve exactly when \(g\) contains an open hemisphere (\Cref{prop:dSd-threshold}), and at most one component of the complement can do so (\Cref{cor:dSd-unique}); at \(d=2\) an open hemisphere is an arc of length \(\pi\), which recovers the threshold above. In higher dimensions a complement component need not be a geodesic ball, so the hemisphere criterion replaces the closed formula \eqref{eq:intro-gap-rule}.

We next study a concrete modular dressing. Each conformal-net region is crossed by its own vacuum modular flow, while all regions share the same auxiliary clock representation. \Cref{thm:net-compatibility-criterion} gives an exact criterion for an inclusion of the concretely represented cores, and shows that it requires the modular flow of the larger region to restrict exactly to that of the smaller region. \Cref{prop:net-compatibility-fails} shows that vacuum correlations rule this out whenever the larger region contains a nonempty arc spacelike to the smaller one. Thus the independently constructed regional crossed products fail isotony. Using a split product state on a fixed separated partition instead, \Cref{prop:split-repair} produces an isotonic dressing on the finite Boolean family generated by its atomic arcs.

We then study the assignment
\begin{align}
  \mathcal A_{\rm model}(a)
  :=\mathcal M(\widetilde A)'\cap\mathcal M_{\rm ref},
  \label{eq:intro-ambient-commutant}
\end{align}
where the commutant is taken in a fixed common ambient representation, as a commutant-based model motivated by the fundamental-complement prescription. Here \(\mathcal M_{\rm ref}\) is the type-\(\mathrm{II}_1\) factor obtained as the positive-clock-spectrum corner of the continuous core of one reference static-patch algebra. Its piecewise definition assigns \(\mathcal A_{\rm model}(a)=\mathcal M_{\rm ref}\) when \(\widetilde A=\varnothing\). The nontrivial statement is \Cref{cor:properness-a1}: the assigned algebra is a \emph{proper} von Neumann subalgebra of \(\mathcal M_{\rm ref}\) whenever the protected gap overlaps the fixed reference arc \(R\). This includes every protected gap of length strictly greater than \(\pi\), as well as all threshold placements except \(R^\perp\). For the explicit family \(a_1=(0,\delta)\), combining properness with the empty-complement branch gives the discontinuity we call activation (\Cref{def:activation}, \Cref{thm:activation}): adding an arbitrarily short antipodal arc moves the assignment from a proper subalgebra to the whole reference algebra.

\subsection{Framework and Scope}
\label{sec:intro-scope}

The geometric theorem applies directly to global de Sitter space. The operator-algebraic analysis uses an abstract chiral conformal net as a controlled model of one light-ray sector, together with an auxiliary \(L^2(\mathbb R)\) clock. The comparison with JT gravity is correspondingly limited: \Cref{app:adm} derives a classical constraint relation between endpoint dilaton data and a one-sided matter charge, while the crossed-product clock remains part of the representation.

The construction leaves several structural questions open. The vacuum-modular dressing fails isotony, and isotony of the ambient-commutant assignment is unresolved when both fundamental complements are nonempty. Factoriality and commutant recovery for \(\mathcal A_{\rm model}(a)\) are also open. A JT realization would require an anomaly-free non-chiral matter theory, a gravitational clock, and coherent regional embeddings. York-time constructions \cite{ParrikarSake,BlommaertChen} offer a possible route, while \cite{SahuGEWAlgebras} treats the wedge-to-algebra map as structural input.

\Cref{sec:geometry} fixes the region conventions and proves the gap-protection rule. \Cref{sec:algebra} builds the common-clock representation in which a single region's modular dressing lives, and isolates in \Cref{sec:stress-tensor} everything that requires a stress-energy tensor. \Cref{sec:isotony} gives the exact criterion for a compatible inclusion of two dressed regions, shows that it fails for their own vacuum modular flows, and analyses the split-state repair. \Cref{sec:activation-model} defines the ambient-commutant model and proves the activation statement. \Cref{sec:discussion} draws out the consequences for the Bousso-Penington--Kaya proposal and for a possible future JT realization. Our notational conventions are as follows: lower-case letters denote spatial regions and capital letters spacetime regions, the spacelike and fundamental complements being the exceptions; a \emph{hologram}, in the sense of \cite{BP,BK}, is a worldtube and is written \(E\); and its intersection with a time slice is the \emph{hologram trace} of \Cref{def:hologram-trace}, written \(e\). The companion paper \cite{Y2H-Page_from_holograms} uses the same conventions.

\section{Geometry of Global \texorpdfstring{dS\(_2\)}{dS2} and Gap Protection}
\label{sec:geometry}

The global de Sitter metric in conformal coordinates is
\begin{align}
  \ed s^2 = \frac{L^2}{\cos^2\eta}\,(-\ed\eta^2+\ed\chi^2),\qquad
  -\frac{\pi}{2}<\eta<\frac{\pi}{2},\quad \chi\sim\chi+2\pi ,
  \label{eq:global-ds-metric}
\end{align}
with Ricci scalar \(R=2/L^2\). The time-reflection-symmetric Cauchy slice is \(\Sigma_0=\{\eta=0\}\), a circle \(S^1\) of circumference \(2\pi L\) whose angular coordinate is \(\chi\). We write \(M\) for the whole spacetime \eqref{eq:global-ds-metric} and \(\mathcal I^\pm\) for its future and past conformal boundaries \(\eta=\pm\pi/2\).

\subsection{Regions and Complements}
\label{sec:region-conventions}

We distinguish spatial from spacetime regions and use separate notation for set, regular-open, and causal complements.

\begin{definition}[Admissible spatial regions and their gaps]
\label{def:admissible-region}
A set \(a\subseteq S^1\) is an \emph{admissible spatial region} if
\begin{enumerate}[label=(\alph*)]
\item \(a=\bigsqcup_{i=1}^{n}I_i\) is a finite union of nonempty open arcs \(I_i\subsetneq S^1\) whose closures \(\overline{I_i}\) are pairwise disjoint; and
\item \(\overline a\ne S^1\), so that \(a\) is a proper subset.
\end{enumerate}
For such \(a\) we write
\begin{align}
  a^\perp:=\intr\bigl(S^1\setminus a\bigr)
  \label{eq:regular-open-complement}
\end{align}
for the \emph{regular-open complement}, and call the connected components \(b_1,\dots,b_m\) of \(a^\perp\) the \emph{gaps} of \(a\), with angular lengths \(\beta_k:=|b_k|\).
\end{definition}

We record the elementary properties used below. First, each \(I_i\) and each \(b_k\) is an open arc, and \(a\) and \(a^\perp\) are \emph{regular open}, that is \(a=\intr\overline a\) and \(a^\perp=\intr\overline{a^\perp}\); consequently
\begin{align}
  a^{\perp\perp}=a .
  \label{eq:double-perp}
\end{align}
Second, \(a^\perp\) is again an admissible spatial region, with gaps \(I_1,\dots,I_n\), and
\begin{align}
  |a|+\sum_{k=1}^{m}\beta_k = 2\pi ,
  \label{eq:length-sum}
\end{align}
because \(S^1\setminus(a\cup a^\perp)\) is the finite set of endpoints and therefore has zero measure. Third, the angular length of a gap is unchanged by adjoining or deleting its endpoints, so the threshold conditions \(\beta_k\ge\pi\) below do not depend on whether gaps are taken open or closed. This is the reason for working with regular open sets: the combinatorics of gap lengths is insensitive to endpoints, whereas domains of dependence and local algebras are not.

\begin{convention}[Complements]
\label{conv:complements}
We use three distinct notions.
\begin{enumerate}[label=(\roman*)]
\item The \emph{set complement} \(S^1\setminus a\), which is closed and contains the endpoints of \(a\).
\item The \emph{regular-open complement} \(a^\perp\) of \eqref{eq:regular-open-complement}. It is the object with respect to which \eqref{eq:double-perp} holds; \(S^1\setminus a\) itself is not involutive under \(\intr(S^1\setminus\,\cdot\,)\) except up to endpoints.
\item The \emph{causal complement} of a spacetime set \(X\subseteq M\),
\begin{align}
  X':=M\setminus\cl\bigl[I(X)\bigr],
  \label{eq:causal-complement}
\end{align}
where \(I(X)\) is the chronological future and past of \(X\) \cite{WaldGR,HawkingEllis}. A \emph{wedge} is a set with \(X=X''\).
\end{enumerate}
For an admissible \(a\) we write \(A:=D(a)\) for the domain of dependence, a spacetime region. Lower-case symbols always denote subsets of \(\Sigma_0\) and upper-case symbols always denote subsets of \(M\).
\end{convention}

Two facts relate (ii) and (iii) and are used repeatedly. For an admissible \(a\), the components of \(a\) have pairwise disjoint closures, so
\begin{align}
  D(a)=\bigsqcup_{i=1}^{n}D(I_i),
  \qquad
  D(a)'=D\bigl(a^\perp\bigr)=\bigsqcup_{k=1}^{m}D(b_k),
  \label{eq:D-and-complements}
\end{align}
and, by \eqref{eq:double-perp}, \(D(a)''=D(a)\); domains of dependence of admissible regions are wedges. We emphasize that \eqref{eq:D-and-complements} requires the regular-open convention: if a single interior endpoint is deleted from \(a\), then \(D(a)\) shrinks by a full causal diamond, and \(D(a)''\ne D(a)\). On the algebraic side the corresponding endpoint is filled in by strong additivity rather than by a geometric identity; see \Cref{ass:regional-net}(iii) and the discussion following it.

\subsection{The Fundamental Complement of One Gap}
\label{sec:fundamental-complement}

We use the definition of \cite{BK}. For a wedge \(X\subseteq M\) let \(\mathscr C(X')\) be the union of all timelike curves contained in \(X'\) whose proper time is infinite in both directions. The \emph{fundamental complement} of \(X\) is the causal completion
\begin{align}
  \widetilde X:=\mathscr C(X')'' ,
  \label{eq:fundamental-complement-definition}
\end{align}
with \(\,'\,\) as in \eqref{eq:causal-complement}. The completion step is essential: it is not enough to decide whether one complete curve exists. In \cite{BK} the gravitating region is itself written as a wedge \(a\) with fundamental complement \(\tilde a\); in the present notation their wedge is our \(A\) and their \(\tilde a\) is our \(\widetilde A\).

Let \(b\subset\Sigma_0\) be an open arc of angular length \(0<\beta<2\pi\). Rotate the circle so that \(b\) is centred at \(\chi=0\), and write
\[
  \operatorname{dist}(\chi,\chi_0)
  :=\min_{n\in\mathbb Z}|\chi-\chi_0+2\pi n|\in[0,\pi]
\]
for circular distance. Tracing the null curves \(\chi\pm\eta=\mathrm{constant}\) from the two endpoints gives
\begin{align}
  D(b)=\left\{(\eta,\chi):
  |\eta|<\frac{\pi}{2},\quad
  \operatorname{dist}(\chi,0)+|\eta|<\frac{\beta}{2}\right\}.
  \label{eq:interval-domain}
\end{align}
\Cref{fig:gap_protection} illustrates the two cases. For \(\beta<\pi\) this is the usual causal diamond, with tips at \(\eta=\pm\beta/2\) inside \(M\). For \(\beta=\pi\) its closure meets each conformal boundary in the single point \((\pm\pi/2,0)\). For \(\beta>\pi\) it meets each boundary in an open angular segment of length \(\beta-\pi\).

\begin{lemma}[Complete curves in an arc domain]
\label{lem:complete-curves}
Let \(C(b)\) denote the set of points lying on past- and future-infinite timelike curves contained in \(D(b)\). Then
\begin{align}
  C(b)=
  \begin{cases}
    \varnothing, & 0<\beta<\pi,\\
    D(b), & \pi\leq\beta<2\pi .
  \end{cases}
  \label{eq:complete-curve-set}
\end{align}
\end{lemma}

\begin{proof}
Along a future-directed timelike curve \(\eta\) is strictly increasing. If the curve remains in a compact substrip \(|\eta|\leq\pi/2-\epsilon\), its proper time cannot diverge: there the conformal factor \(L/\cos\eta\) is bounded, so the Lorentzian length is bounded by a constant times the coordinate-time interval, which is at most \(\pi-2\epsilon\). Hence a curve infinite in both directions must approach both \(\mathcal I^-\) and
\(\mathcal I^+\).

If \(\beta<\pi\), \eqref{eq:interval-domain} forces \(|\eta|<\beta/2<\pi/2\) on \(D(b)\). No curve in \(D(b)\) reaches either conformal boundary, so \(C(b)=\varnothing\).

Now let \(\beta\geq\pi\) and fix \(p=(\eta_p,\chi_p)\in D(b)\). Work in the lift \(x\in\mathbb R\) of \(\chi\) with \(x_p\in(-\pi,\pi]\), so that \eqref{eq:interval-domain} reads \(|x_p|+|\eta_p|<\beta/2\) and the lifted region \(\{|x|+|\eta|<\beta/2\}\cap\{|\eta|<\pi/2\}\) is convex. Suppose first \(\beta>\pi\). The intervals
\begin{align*}
 J_+&=\left(x_p-\Bigl(\tfrac{\pi}{2}-\eta_p\Bigr),
               x_p+\Bigl(\tfrac{\pi}{2}-\eta_p\Bigr)\right)
       \cap\left(-\tfrac{\beta-\pi}{2},\tfrac{\beta-\pi}{2}\right),\\
 J_-&=\left(x_p-\Bigl(\tfrac{\pi}{2}+\eta_p\Bigr),
               x_p+\Bigl(\tfrac{\pi}{2}+\eta_p\Bigr)\right)
       \cap\left(-\tfrac{\beta-\pi}{2},\tfrac{\beta-\pi}{2}\right)
\end{align*}
are nonempty, which is exactly the strict inequality \(|x_p|+|\eta_p|<\beta/2\). Pick \(x_\pm\in J_\pm\). The straight segments from \((-\pi/2,x_-)\) to \(p\) and from \(p\) to \((\pi/2,x_+)\) have slopes of magnitude strictly smaller than one, and they remain in the lifted region: that region is open and convex, \(p\) lies in it, and both endpoints lie in its closure, so the half-open segments are contained in it. After smoothing at \(p\) they give a timelike curve in \(D(b)\) through \(p\); its proper time diverges at both ends because \(\int^{\pm\pi/2} \ed\eta/\cos\eta\) diverges.

At the threshold \(\beta=\pi\) the open intervals \(\bigl(-\tfrac{\beta-\pi}{2},\tfrac{\beta-\pi}{2}\bigr)\) are empty, so \(J_\pm=\varnothing\) and the construction above does not apply literally. Take \(x_\pm=0\) instead. The domain inequality gives \(|x_p|<\pi/2-|\eta_p|\le\pi/2\mp\eta_p\), so the two straight segments joining \(p\) to \((\pm\pi/2,0)\) again have slope of magnitude strictly smaller than one; and \((\pm\pi/2,0)\) lies in the closure of the lifted region because \(|0|+\pi/2=\beta/2\), so convexity keeps the half-open segments inside. The threshold case needs no continuity or limiting argument. Hence every point of \(D(b)\) lies on a complete curve for all \(\beta\geq\pi\).
\end{proof}

\begin{proposition}[Causal completion of one gap]
\label{prop:single-gap-completion}
For \(0<\beta<2\pi\),
\begin{align}
 C(b)'&=
 \begin{cases}
   M,&\beta<\pi,\\
   D(b)',&\beta\geq\pi,
 \end{cases}
 &
 C(b)''&=
 \begin{cases}
   \varnothing,&\beta<\pi,\\
   D(b),&\beta\geq\pi .
 \end{cases}
 \label{eq:single-gap-completion}
\end{align}
\end{proposition}

\begin{proof}
For \(\beta<\pi\), \Cref{lem:complete-curves} gives \(C(b)=\varnothing\), and
\(\varnothing'=M\), \(M'=\varnothing\).
For \(\beta\geq\pi\) it gives \(C(b)=D(b)\). Since \(b\) is an open arc it is regular open, so \eqref{eq:D-and-complements} applies with \(a=b\):
\begin{align}
 D(b)'=D\bigl(b^\perp\bigr)=
 \left\{(\eta,\chi):
 \operatorname{dist}(\chi,\pi)+|\eta|<\frac{2\pi-\beta}{2}\right\},
 \label{eq:gap-chronological-complement}
\end{align}
the diamond of the complementary arc centred at \(\chi=\pi\), and \(D(b)''=D(b^{\perp\perp})=D(b)\). In particular, at \(\beta=\pi\) the complete curves that limit to the two single boundary points causally generate the entire half-circle diamond.
\end{proof}

\subsection{Gap Protection for Finite Unions of Arcs}
\label{sec:gap-protection}

\begin{theorem}[Global-\(\mathrm{dS}_2\) gap theorem]
\label{thm:gap-protection}
Let \(a\subset S^1\) be an admissible spatial region (\Cref{def:admissible-region}) with gaps \(b_1,\dots,b_m\) of angular lengths \(\beta_k\), and let \(A=D(a)\). Then the spacetime fundamental complement of
\(A\) is
\begin{align}
  \widetilde A = \bigcup_{\beta_k\ge \pi} D(b_k) ,
  \label{eq:gap-protection-result}
\end{align}
and this union has at most one term. A complementary gap is therefore protected precisely when its angular length is at least \(\pi\), the equality case included. Moreover, writing \(\widetilde a:=\widetilde A\cap\Sigma_0\),
\begin{align}
  \widetilde a=\bigcup_{\beta_k\ge\pi}b_k,
  \qquad
  \widetilde A = D(\widetilde a).
  \label{eq:spatial-trace-identity}
\end{align}
\end{theorem}

\begin{proof}
By \eqref{eq:D-and-complements} the causal complement of \(A\) is the disjoint union \(A'=\bigsqcup_{k=1}^{m}D(b_k)\), and the \(D(b_k)\) are distinct connected components of \(A'\). The image of a timelike curve is connected, so a complete timelike curve contained in \(A'\) lies in exactly one \(D(b_k)\). \Cref{lem:complete-curves} therefore gives
\[
 \mathscr C(A')=\bigsqcup_{\beta_k\geq\pi}D(b_k).
\]

Next, at most one gap can be protected. By \eqref{eq:length-sum} and \(|a|>0\) we have \(\sum_k\beta_k<2\pi\), so two distinct gaps of angular length at least \(\pi\) are impossible. Hence \(\mathscr C(A')\) is either empty or a single \(D(b_{k})\), and taking the causal completion cannot join different components. By \Cref{prop:single-gap-completion}, \(\widetilde A=\mathscr C(A')''\) is therefore respectively \(\varnothing\) or \(D(b_{k})\), which is \eqref{eq:gap-protection-result}.

For \eqref{eq:spatial-trace-identity}, setting \(\eta=0\) in \eqref{eq:interval-domain} gives \(D(b_k)\cap\Sigma_0=b_k\). If no gap is protected both sides of \eqref{eq:spatial-trace-identity} are empty, using \(D(\varnothing)=\varnothing\). If \(b_*\) is the unique protected gap then \(\widetilde A=D(b_*)\) and \(\widetilde a=b_*\), so
\(D(\widetilde a)=D(b_*)=\widetilde A\).
\end{proof}

\begin{figure}[!tb]
{\sffamily\sansmath
\centering
\makebox[\textwidth][c]{%
\resizebox{1.02\linewidth}{!}{%
\begin{tikzpicture}[x=1.0cm, y=1.25cm, >=Stealth]

\begin{scope}[shift={(0,0)}]
    \fill[fill=figred!12]
        (1.257, 0) -- (2.827, 1.571) -- (3.456, 1.571) -- (5.027, 0) -- (3.456, -1.571) -- (2.827, -1.571) -- cycle;
    \draw[thick, double] (0, 1.571) -- (6.283, 1.571) node[midway, above=4pt, font=\footnotesize] {$\mathcal{I}^+\ (\eta = +\pi/2)$};
    \draw[thick, double] (0, -1.571) -- (6.283, -1.571) node[midway, below=4pt, font=\footnotesize] {$\mathcal{I}^-\ (\eta = -\pi/2)$};
    \draw[thick, dashed] (0, -1.571) -- (0, 1.571);
    \draw[thick, dashed] (6.283, -1.571) -- (6.283, 1.571);
    \node[left=3pt, font=\footnotesize] at (0, -1.0) {$\chi=0$};
    \node[anchor=east, font=\footnotesize] at (6.20, -1.0) {$\chi=2\pi$};
    \draw[gray!50, thick] (0,0) -- (6.283,0);
    \node[left=4pt, font=\footnotesize\bfseries, gray!80!black] at (0,0) {$\Sigma_0$};
    \draw[line width=2.5pt, gapgreen] (1.257, 0) -- (5.027, 0);
    \draw[gapgreen, fill=white, line width=1pt] (1.257, 0) circle (2pt);
    \draw[gapgreen, fill=white, line width=1pt] (5.027, 0) circle (2pt);
    \node[above=2pt, font=\footnotesize, inner sep=1.5pt] at (2.1, 0.05) {$b$};
    \draw[thick, figred] (1.257, 0) -- (2.827, 1.571);
    \draw[thick, figred] (1.257, 0) -- (2.827, -1.571);
    \draw[thick, figred] (5.027, 0) -- (3.456, 1.571);
    \draw[thick, figred] (5.027, 0) -- (3.456, -1.571);
    \draw[very thick, figred, <->] (3.142, -1.48) -- (3.142, 1.48);
    \node[
      below=16pt,
      inner sep=0pt
    ] at (3.142,-1.8)
      {\begin{minipage}{7cm}
      \centering
      \footnotesize\sffamily\sansmath
      \textbf{(a)} Protected gap:
      $\beta=1.2\pi\ge\pi$\\[1pt]
      \textcolor{figred}{%
        $\widetilde{A}=D(b)\neq\varnothing$%
      }
      \end{minipage}};
\end{scope}

\begin{scope}[shift={(7.4,0)}]
    \fill[fill=figblue!12]
        (2.199, 0) -- (3.142, 0.942) -- (4.084, 0) -- (3.142, -0.942) -- cycle;
    \draw[thick, double] (0, 1.571) -- (6.283, 1.571) node[midway, above=4pt, font=\footnotesize] {$\mathcal{I}^+\ (\eta = +\pi/2)$};
    \draw[thick, double] (0, -1.571) -- (6.283, -1.571) node[midway, below=4pt, font=\footnotesize] {$\mathcal{I}^-\ (\eta = -\pi/2)$};
    \draw[thick, dashed] (0, -1.571) -- (0, 1.571);
    \draw[thick, dashed] (6.283, -1.571) -- (6.283, 1.571);
    \node[right=3pt, font=\footnotesize, text opacity=0] at (6.283, -1.2) {$\chi=2\pi$};
    \draw[gray!50, thick] (0,0) -- (6.283,0);
    \node[left=4pt, font=\footnotesize\bfseries, gray!80!black] at (0,0) {$\Sigma_0$};
    \draw[line width=2.5pt, gapgreen] (2.199, 0) -- (4.084, 0);
    \draw[gapgreen, fill=white, line width=1pt] (2.199, 0) circle (2pt);
    \draw[gapgreen, fill=white, line width=1pt] (4.084, 0) circle (2pt);
    \node[above=2pt, font=\footnotesize, inner sep=1.5pt] at (2.7, 0.0) {$b$};
    \draw[thick, figblue] (2.199, 0) -- (3.142, 0.942);
    \draw[thick, figblue] (2.199, 0) -- (3.142, -0.942);
    \draw[thick, figblue] (4.084, 0) -- (3.142, 0.942);
    \draw[thick, figblue] (4.084, 0) -- (3.142, -0.942);
    \fill[figblue] (3.142, 0.942) circle (2pt);
    \fill[figblue] (3.142, -0.942) circle (2pt);
    \draw[thick, <->, figblue] (3.142, -0.88) -- (3.142, 0.88);
    \node[anchor=west, font=\footnotesize, figblue, fill=white, inner sep=2pt, rounded corners=2pt]
        at (3.79, -0.45) {$\Delta\eta = \beta < \pi$};
    \node[
      below=16pt,
      inner sep=0pt
      ] at (3.142,-1.8)
        {\begin{minipage}{7cm}
        \centering
        \footnotesize\sffamily\sansmath
        \textbf{(b)} Unprotected gap:
        $\beta=0.6\pi<\pi$\\[1pt]
        \textcolor{figblue}{%
          $\widetilde{A}=\varnothing$%
        }
      \end{minipage}};
\end{scope}

\end{tikzpicture}%
}%
}
\caption{Penrose diagrams of global $\mathrm{dS}_2$ illustrating geometric gap protection (\Cref{thm:gap-protection}). The Cauchy slice $\Sigma_0=S^1$ sits at $\eta=0$ and light rays propagate at $45^\circ$; open circles indicate that the gaps are open arcs. \textbf{(a)}~A complementary gap $b$ of angular length $\beta=1.2\pi\ge\pi$ has a domain of dependence $D(b)$ of full conformal height $\Delta\eta=\pi$ that meets both conformal boundaries $\mathcal I^\pm$. It contains timelike curves of infinite proper time in both directions, one of which is drawn as the red double arrow, so $D(b)$ is the spacetime fundamental complement $\widetilde A$ of $A=D(b^\perp)$. \textbf{(b)}~A gap of angular length $\beta=0.6\pi<\pi$ has a diamond whose null boundaries meet in the bulk at $\eta=\pm\beta/2$ before reaching $\mathcal I^\pm$; no curve inside reaches a conformal boundary and the fundamental complement is empty.}
}
\label{fig:gap_protection}
\end{figure}
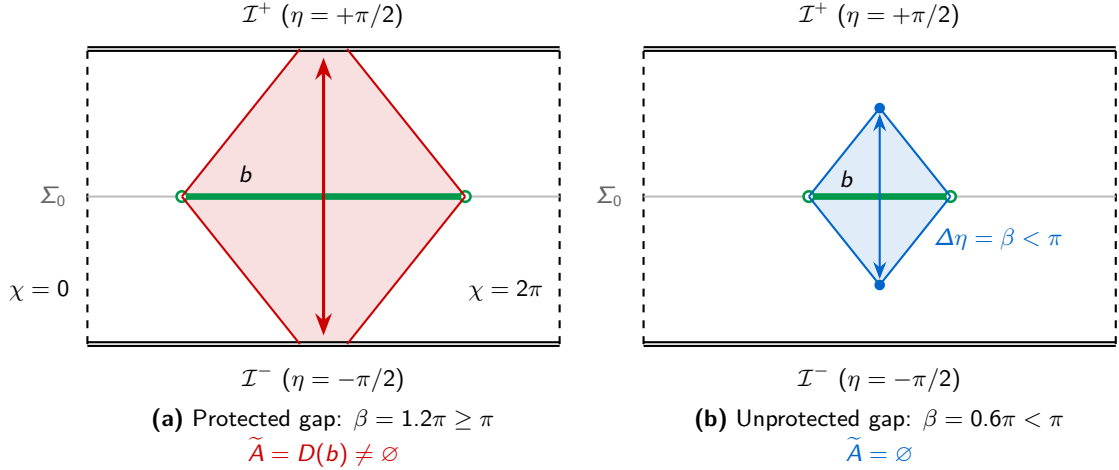

\begin{remark}[The spatial trace]
\label{rem:spatial-trace-caution}
A disconnected spacetime region need not equal the domain of dependence of its spatial trace, so \eqref{eq:spatial-trace-identity} rests on \Cref{thm:gap-protection} having already reduced \(\widetilde A\) to a single arc. We write \(\mathcal M(\widetilde A)\) for the algebra of the fundamental complement in \Cref{sec:activation-model}, and use \(\mathcal M(\widetilde A)=\mathcal M(D(\widetilde a))\) only where \eqref{eq:spatial-trace-identity} applies.
\end{remark}

We call gaps with \(\beta_k\geq\pi\) \emph{protected} and gaps with \(\beta_k<\pi\) \emph{unprotected}. Whether an unprotected gap enters a state-dependent hologram \(E\) is a separate, entropic question, discussed in \Cref{sec:refined-wedge-empty}.

\FloatBarrier
\subsection{The Threshold in Higher Dimensions}
\label{sec:higher-dimensions}

The half-circle threshold is not special to two dimensions. No proof below uses this subsection; we include it because it identifies what the threshold of \Cref{thm:gap-protection} really measures.

Let \(d\ge2\) and write global \(\mathrm{dS}_d\) in conformal coordinates as
\begin{align}
  \ed s^2=\frac{L^2}{\cos^2\eta}\bigl(-\ed\eta^2+\ed\Omega_{d-1}^2\bigr),
  \qquad -\frac{\pi}{2}<\eta<\frac{\pi}{2},
  \label{eq:dSd-metric}
\end{align}
with \(\ed\Omega_{d-1}^2\) the round metric of unit radius on \(\Sigma_0=S^{d-1}\). Let \(\operatorname{dist}\) be the induced geodesic distance, which takes values in \([0,\pi]\) and for which the antipode of \(p\) is the unique point at distance \(\pi\) from it; at \(d=2\) this is the circular distance of \Cref{sec:fundamental-complement}. Write \(\mathcal B_\theta(p):=\{q:\operatorname{dist}(q,p)<\theta\}\) for the open geodesic ball, and call \(\mathcal B_{\pi/2}(p)\) an \emph{open hemisphere}. In this subsection \(p,q\) denote points of \(S^{d-1}\); the unrelated clock coordinate \(q\) of \Cref{sec:reference-crossed-product} appears nowhere in \Cref{sec:geometry}. At \(d=2\) an open ball of radius \(\theta\) is an open arc of angular length \(2\theta\), and an open hemisphere is an open arc of length \(\pi\).

\begin{lemma}[Domain of dependence on the round sphere]
\label{lem:dSd-domain}
For open \(g\subseteq S^{d-1}\),
\begin{align}
  D(g)=\Bigl\{(\eta,q):|\eta|<\tfrac{\pi}{2},\
  \overline{\mathcal B_{|\eta|}(q)}\subseteq g\Bigr\}.
  \label{eq:dSd-domain}
\end{align}
In particular, for \(0<\theta<\pi\),
\begin{align}
  D\bigl(\mathcal B_\theta(p)\bigr)
  =\Bigl\{(\eta,q):|\eta|<\tfrac{\pi}{2},\
  |\eta|+\operatorname{dist}(q,p)<\theta\Bigr\}.
  \label{eq:dSd-ball-domain}
\end{align}
\end{lemma}

\begin{proof}
The metric \eqref{eq:dSd-metric} is conformal to the static cylinder \(-\ed\eta^2+\ed\Omega_{d-1}^2\), so the two have the same causal curves. Along a causal curve of the cylinder the sphere component has speed at most one, and geodesic distance is \(1\)-Lipschitz along curves, so a causal curve joining \((\eta_0,q_0)\) to \((\eta_1,q_1)\) satisfies
\begin{align}
  \operatorname{dist}(q_0,q_1)\le|\eta_1-\eta_0| .
  \label{eq:dSd-lipschitz}
\end{align}

Fix \(\eta>0\); for \(\eta=0\) the assertion is just that \(q\in g\). Global \(\mathrm{dS}_d\) is globally hyperbolic with Cauchy surface \(\Sigma_0\), so every past-inextendible causal curve from \((\eta,q)\) crosses \(\Sigma_0\), and by \eqref{eq:dSd-lipschitz} it does so inside \(\overline{\mathcal B_\eta(q)}\). Conversely every \(q'\) in that closed ball is a crossing point: traverse a minimizing geodesic of \(S^{d-1}\) from \(q\) to \(q'\) at the constant sphere-speed \(\operatorname{dist}(q,q')/\eta\le1\), which gives a causal curve from \((\eta,q)\) to \((0,q')\), and continue it past-inextendibly below \(\Sigma_0\). Hence every past-inextendible causal curve from \((\eta,q)\) meets \(g\) exactly when \(\overline{\mathcal B_\eta(q)}\subseteq g\). The same argument with time reversed gives \(D^-\), and together they give \eqref{eq:dSd-domain}.

For \eqref{eq:dSd-ball-domain}, the triangle inequality gives \(\operatorname{dist}(x,p)\le\operatorname{dist}(p,q)+\eta\) for every \(x\in\overline{\mathcal B_\eta(q)}\), so the maximum of \(\operatorname{dist}(\cdot,p)\) over that compact set is at most \(\min\{\pi,\operatorname{dist}(p,q)+\eta\}\); it equals that value, attained by continuing a minimizing geodesic from \(p\) through \(q\) a further distance \(\eta\) when \(\operatorname{dist}(p,q)+\eta\le\pi\), and at the antipode of \(p\) otherwise. Since \(\theta<\pi\), the containment \(\overline{\mathcal B_\eta(q)}\subseteq\mathcal B_\theta(p)\) therefore holds if and only if \(\operatorname{dist}(p,q)+\eta<\theta\).
\end{proof}

\begin{proposition}[The threshold is a hemisphere in every dimension]
\label{prop:dSd-threshold}
Let \(d\ge2\) and let \(g\subseteq S^{d-1}\) be open. Then \(D(g)\) contains a timelike curve of infinite proper time in both directions if and only if \(g\) contains an open hemisphere.
\end{proposition}

\begin{proof}
Suppose \(\mathcal B_{\pi/2}(p)\subseteq g\). Isotony of \(D\) gives \(D(\mathcal B_{\pi/2}(p))\subseteq D(g)\), and by \eqref{eq:dSd-ball-domain} the worldline \(\{(\eta,p):|\eta|<\pi/2\}\) lies in \(D(\mathcal B_{\pi/2}(p))\). Its proper time is \(L\int_{-\pi/2}^{\pi/2}\ed\eta/\cos\eta\), which diverges at both ends.

Conversely, let \(\gamma\subseteq D(g)\) be timelike with infinite proper time in both directions. On a substrip \(|\eta|\le\pi/2-\epsilon\) the conformal factor \(L/\cos\eta\) is at most \(L/\sin\epsilon\), and \(\eta\) is monotonic along a timelike curve, so the proper time \(\gamma\) accumulates inside that substrip is at most \(\pi L/\sin\epsilon\). Being infinite in both directions, \(\gamma\) leaves every such substrip, and in particular has points \((\eta_n,q_n)\) with \(\eta_n\uparrow\pi/2\). By \eqref{eq:dSd-domain}, \(\overline{\mathcal B_{\eta_n}(q_n)}\subseteq g\). Since \(S^{d-1}\) is compact we may pass to a subsequence with \(q_n\to p\). Fix \(\theta<\pi/2\). For \(n\) large enough that \(\eta_n\ge\theta+\operatorname{dist}(p,q_n)\), the triangle inequality gives
\(\overline{\mathcal B_\theta(p)}\subseteq
\overline{\mathcal B_{\eta_n}(q_n)}\subseteq g\). As \(\theta<\pi/2\) was arbitrary,
\(\mathcal B_{\pi/2}(p)=\bigcup_{\theta<\pi/2}\overline{\mathcal B_\theta(p)}
\subseteq g\).
\end{proof}

\begin{corollary}[At most one protected component]
\label{cor:dSd-unique}
Let \(a\subseteq S^{d-1}\) be open with \(|a|>0\). Then at most one connected component of \(a^\perp=\intr(S^{d-1}\setminus a)\) contains an open hemisphere.
\end{corollary}

\begin{proof}
Two disjoint open hemispheres \(\mathcal B_{\pi/2}(p)\) and \(\mathcal B_{\pi/2}(p')\) force \(\operatorname{dist}(p,p')=\pi\): otherwise the midpoint of a minimizing geodesic from \(p\) to \(p'\) lies at distance \(\operatorname{dist}(p,p')/2<\pi/2\) from both and belongs to both balls. So \(p'\) is the antipode of \(p\), and \(\mathcal B_{\pi/2}(p)\cup\mathcal B_{\pi/2}(-p)=S^{d-1}\setminus E\) with \(E=\{\operatorname{dist}(\cdot,p)=\pi/2\}\) an equatorial great sphere, of measure zero. Both hemispheres lie in \(a^\perp\), which is disjoint from \(a\), so \(a\subseteq E\) and \(|a|=0\).
\end{proof}

\begin{remark}[Higher-dimensional extension]
\label{rem:dSd-scope}
At \(d=2\) a component of \(a^\perp\) is an arc, of angular length \(\beta\) and hence a ball of radius \(\beta/2\), so \Cref{prop:dSd-threshold} reads \(\beta\ge\pi\): the threshold of \Cref{thm:gap-protection}, equality included. \Cref{cor:dSd-unique} likewise reduces to the length count in the proof of that theorem.

For \(d\ge3\), a component of \(a^\perp\) need not be a geodesic ball. \Cref{prop:dSd-threshold} determines whether it carries a complete curve, but the causal completion \(\mathscr C(A')''\) may be smaller than \(D(g)\). Thus the higher-dimensional result is a threshold criterion rather than a closed formula for \(\widetilde A\).
\end{remark}

\section[Regional Modular Dressings and the Common-Clock Representation]%
{Regional Modular Dressings and the Common-Clock Representation}
\label{sec:algebra}

\subsection{Modular Conventions and the Static Normalization}
\label{sec:modular-conventions}

We place a chiral conformal net on the circle \(\Sigma_0=S^1\) in its vacuum state, denoted \(\Omega_{\rm BD}:=\ket{\rm BD}\) in the de Sitter application \cite{BunchDavies}. For every admissible region \(u\) for which \(\Omega_{\rm BD}\) is cyclic and separating, \(\Delta_u\) denotes the Tomita modular operator of \((\mathcal A_{\rm CFT}(u),\Omega_{\rm BD})\) and
\begin{align}
 K_u:=-\log\Delta_u,\qquad
 \sigma_t^u(x):=\Delta_u^{it}x\Delta_u^{-it}
               =e^{-itK_u}x\,e^{itK_u}.
 \label{eq:modular-convention}
\end{align}
Thus \(t\) is dimensionless mathematical modular time, and \(\Omega_{\rm BD}\) is a \(1\)-KMS state for \(\sigma^u\).

The reference region is the half-circle
\begin{align}
  R:=\bigl(-\tfrac{\pi}{2},\tfrac{\pi}{2}\bigr)\subset S^1,
  \label{eq:reference-region}
\end{align}
whose domain of dependence \(D(R)\) is a static patch of \eqref{eq:global-ds-metric}. Let \(\delta^R_\tau\) be the one-parameter group of M\"obius transformations of \(S^1\) that preserves \(R\) and fixes its two endpoints, parametrized so that on \(\Sigma_0\) it is the flow of the de Sitter Killing field
\begin{align}
  \xi=\cos\chi\,\partial_\eta ,
  \label{eq:killing-field}
\end{align}
which is timelike and future-directed throughout \(D(R)\), past-directed in the opposite patch, and vanishes at \(\chi=\pm\pi/2\). We write \(H_{\rm stat}^{\rm ch}\) for the self-adjoint generator of the unitary implementation of \(\delta^R\) in the vacuum representation of the abstract chiral net,
\begin{align}
  \alpha_\tau^{\rm stat}
  :=\operatorname{Ad}\bigl(e^{i\tau H_{\rm stat}^{\rm ch}}\bigr),
  \qquad
  H_{\rm stat}^{\rm ch}\,\Omega_{\rm BD}=0 ,
  \label{eq:hstat-definition}
\end{align}
with \(\tau\) the \emph{dimensionless} static Killing parameter; the central observer's proper time is \(T=L\tau\), so the dimensionful generator is \(H_{\rm stat}^{\rm ch,proper}=H_{\rm stat}^{\rm ch}/L\).

The operator \(H_{\rm stat}^{\rm ch}\) is the group-theoretic generator in the chiral vacuum representation. The stress-tensor charge used in the JT comparison belongs to the separate full-CFT realization introduced in \Cref{sec:stress-tensor}.

\begin{lemma}[Static modular normalization]
\label{lem:static-modular-relation}
For a M\"obius-covariant net with positive energy and vacuum
\(\Omega_{\rm BD}\),
\begin{align}
  \sigma_t^R=\alpha_{-2\pi t}^{\rm stat},
  \qquad \Delta_R^{it}=e^{-2\pi itH_{\rm stat}^{\rm ch}},
  \qquad K_R=2\pi H_{\rm stat}^{\rm ch}
  =2\pi L\,H_{\rm stat}^{\rm ch,proper} .
  \label{eq:static-modular-relation}
\end{align}
\end{lemma}

\begin{proof}
The Bisognano--Wichmann property for M\"obius-covariant nets on \(S^1\) \cite[Thm.~2.3(ii)]{BrunettiGuidoLongo1993} states that the modular unitaries of an interval algebra in the vacuum state are exactly the unitaries implementing the one-parameter M\"obius subgroup preserving that interval, run at \(-2\pi\) times the modular parameter. For \(R\) that subgroup is \(\delta^R\) with generator \(H_{\rm stat}^{\rm ch}\), so \(\Delta_R^{it}=e^{-2\pi itH_{\rm stat}^{\rm ch}}\), which is \(\sigma^R_t=\alpha^{\rm stat}_{-2\pi t}\). Comparing with \(\Delta_R^{it}=e^{-itK_R}\) from \eqref{eq:modular-convention} and using that both one-parameter unitary groups fix \(\Omega_{\rm BD}\), uniqueness of the Stone generator gives \(K_R=2\pi H_{\rm stat}^{\rm ch}\).

The factor \(-2\pi\) is equivalent to the more familiar statement that \(\omega_{\rm BD}|_{\mathcal A_{\rm CFT}(R)}\) is KMS at inverse temperature \(\beta_{\rm stat}=2\pi\) for \(\alpha^{\rm stat}\), and the two agree: a \(\beta_{\rm stat}\)-KMS state for \(\alpha_\tau=\operatorname{Ad}(e^{i\tau H})\) has modular group \(\sigma_t=\alpha_{-\beta_{\rm stat}t}\), since in the Gibbs model \(\rho\propto e^{-\beta_{\rm stat}H}\) gives \(\rho^{it}x\rho^{-it}=e^{-i\beta_{\rm stat}tH}x\,e^{i\beta_{\rm stat}tH}\).
\end{proof}

\begin{remark}[Orientation and normalization]
\label{rem:sign-check}
The factor \(2\pi\) and the sign in \eqref{eq:static-modular-relation} follow from three conventions. (i) The Tomita convention \(K_u=-\log\Delta_u\) in \eqref{eq:modular-convention}, for which \(\Omega_{\rm BD}\) is \(1\)-KMS in \(t\). (ii) The sign of \eqref{eq:hstat-definition}, with \(H_{\rm stat}^{\rm ch}\) generating the flow of \eqref{eq:killing-field}, future-directed in \(D(R)\). (iii) The value \(\beta_{\rm stat}=2\pi\), which in dimensionless \(\tau\) is the de Sitter temperature \(T_{\rm dS}=1/(2\pi L)\) measured in the central observer's proper time \(T=L\tau\). Reversing (ii) --- taking the generator past-directed in \(D(R)\) --- flips the sign in \eqref{eq:static-modular-relation}, and this is the only orientation freedom involved. As a geometric cross-check, choose a M\"obius coordinate \(x\) mapping \(R\) to \((0,\infty)\) and oriented so that \(\alpha^{\rm stat}_\tau\) acts as \(x\mapsto e^{-\tau}x\). Then \eqref{eq:static-modular-relation} says that \(\sigma^R_t\) acts as the dilation \(x\mapsto e^{2\pi t}x\), which is the statement used in \Cref{lem:geometric-modular-flow}. Equation \eqref{eq:static-modular-relation} relates the two flow parameters in the vacuum representation.
\end{remark}

\subsection{The Abstract Chiral Net}
\label{sec:abstract-net}

\begin{assumption}[Algebraic CFT framework]
\label{ass:regional-net}
For all operator-algebraic results, \(\mathcal A_{\rm CFT}\) is an abstract chiral conformal net on \(S^1\): it is isotonic, local, irreducible, and M\"obius covariant with positive energy and an invariant vacuum \(\Omega_{\rm BD}\), and \(\mathcal A_{\rm CFT}(R)\ne\mathbb C\id\). We use the following three properties.
\begin{enumerate}[label=(\roman*)]
\item Every nonempty proper open-arc algebra is a type-\(\mathrm{III}_1\) factor.
\item \(\Omega_{\rm BD}\) has the Reeh--Schlieder property.
\item The net is strongly additive and has the split property for arcs with disjoint closures.
\end{enumerate}
Items (i) and (ii) follow from the preceding axioms. Factoriality of arc algebras follows from irreducibility, and ergodicity of the vacuum modular (dilation) group then forces type \(\mathrm{III}_1\) unless the theory is one-dimensional \cite[Part~II, Prop.~6.2.9]{LongoLectureNotes}; see also \cite{GabbianiFrohlich1993} for the type-\(\mathrm{III}_1\) classification of local algebras in chiral conformal field theory. The Reeh--Schlieder property and additivity likewise follow from the stated axioms \cite[Part~II, Thm.~6.2.2 and Cor.~6.2.6]{LongoLectureNotes}. The Bisognano--Wichmann property used in \Cref{lem:static-modular-relation} and Haag duality on \(S^1\) for a single arc are proved from M\"obius covariance and positivity of the energy alone in \cite{BrunettiGuidoLongo1993}, with no stress-energy tensor and no diffeomorphism covariance.

Item (iii) is an additional hypothesis. Strong additivity is equivalent to Haag duality on the once-punctured circle \cite{GuidoLongoWiesbrock1998} and is not implied by Haag duality on \(S^1\); the split property is an independent nuclearity-type condition, which holds whenever \(e^{-sL_0}\) is trace class for every \(s>0\) \cite{BuchholzDAntoniLongo2007}, but which does not follow from the basic axioms. Complete rationality and finiteness of the \(\mu\)-index are not required.
\end{assumption}

The class of regions and the role of these assumptions are as follows. For a single open arc, Haag duality and geometric modular flow are consequences of the M\"obius vacuum representation \cite{BrunettiGuidoLongo1993}. For an admissible region \(u=\bigsqcup_{j=1}^n I_j\) with pairwise disjoint arc closures, the split property gives an isomorphism
\begin{align}
 \bigvee_{j=1}^{n}\mathcal A_{\rm CFT}(I_j)
 \;\xrightarrow{\ \sim\ }\;
 \overline{\bigotimes}_{j=1}^{n}\mathcal A_{\rm CFT}(I_j),
 \label{eq:multi-interval-split}
\end{align}
discussed in detail as \(\Theta_S\) in \Cref{sec:split-repair}. If two arcs have touching closures and differ from their union by a single point, strong additivity fills that point and merges them, after which \eqref{eq:multi-interval-split} applies; this is the precise algebraic role of strong additivity announced after \eqref{eq:D-and-complements}. Combining \Cref{ass:regional-net}(i), stability of type \(\mathrm{III}_1\) under finite tensor products, and \eqref{eq:multi-interval-split}, every admissible regional algebra \(\mathcal A_{\rm CFT}(u)\) is a type-\(\mathrm{III}_1\) factor \cite{Takesaki}. Without strong additivity and the split property this conclusion for disconnected \(u\) does not follow from the single-arc classification.

Throughout the regional construction both \(u\) and \(u^\perp\) contain a nonempty open arc. The vacuum is then cyclic for \(\mathcal A_{\rm CFT}(u)\), because that algebra contains an arc algebra for which it is cyclic, and separating, because \(u^\perp\) contains an arc algebra for which the vacuum is cyclic and locality places that algebra in \(\mathcal A_{\rm CFT}(u)'\). The empty region and the whole circle are excluded; in particular \(\mathcal A_{\rm CFT}(S^1)=B(\mathcal H_{\rm CFT})\) is a type-\(\mathrm I_\infty\) factor and is not covered by \Cref{ass:regional-net}(i).

For every region in this class, cyclicity and separatingness define \(\Delta_u\), \(K_u\), and \(\sigma^u\). The disconnected-region arguments use this Tomita flow without requiring it to act geometrically.

\begin{lemma}[Conformal covariance of arc modular Hamiltonians]
\label{lem:conformal-covariance}
Let \(b_1,b_2\subsetneq S^1\) be open arcs, of any length and position, and let \(K_{b_1},K_{b_2}\) be the modular Hamiltonians of \(\ket{\rm BD}\) restricted to \(\mathcal A_{\rm CFT}(b_1),\mathcal A_{\rm CFT}(b_2)\). Then \(K_{b_1}\) and \(K_{b_2}\) are unitarily equivalent, implemented by a M\"obius unitary.
\end{lemma}
\begin{proof}
The M\"obius group acts on \(S^1\) transitively on ordered pairs of distinct points and preserves orientation, so an element carrying the ordered endpoints of \(b_1\) to those of \(b_2\) carries \(b_1\) onto \(b_2\) and not onto the complementary arc: there is \(g\) with \(g(b_1)=b_2\). The vacuum is M\"obius invariant, so the implementing unitary \(U_g\) fixes \(\ket{\rm BD}\) and \(U_g\mathcal A_{\rm CFT}(b_1)U_g^{-1}=\mathcal A_{\rm CFT}(b_2)\). Uniqueness of modular objects then gives
\(U_gK_{b_1}U_g^{-1}=K_{b_2}\).
\end{proof}

\begin{lemma}[Geometric modular flow for a single arc]
\label{lem:geometric-modular-flow}
For a single open arc \(u\subsetneq S^1\), the modular flow \(\sigma_t^u=\operatorname{Ad}(e^{-itK_u})\) acts geometrically: there is a one-parameter group \(g_t^u\) of M\"obius transformations preserving \(u\) and fixing \(\partial u\) with \(\sigma_t^u(\mathcal A_{\rm CFT}(J))=\mathcal A_{\rm CFT}(g_t^u(J))\) for every arc \(J\subseteq u\). In a M\"obius coordinate mapping \(u\) to \((0,\infty)\), oriented as in \Cref{rem:sign-check}, \(g_t^u\) is
\(x\mapsto e^{2\pi t}x\).
\end{lemma}
\begin{proof}
This is the Bisognano--Wichmann property for a M\"obius-covariant net on \(S^1\) \cite[Thm.~2.3(ii)]{BrunettiGuidoLongo1993} together with \Cref{lem:conformal-covariance}, which transports the statement from \(R\) to an arbitrary arc. The sign and the factor \(2\pi\) are those of \Cref{lem:static-modular-relation} and \Cref{rem:sign-check}: Tomita modular time has KMS inverse temperature \(1\), while the associated geometric flow has inverse temperature \(2\pi\).
\end{proof}

\begin{remark}[Disconnected modular flow]
\label{rem:disconnected-modular-flow}
For an admissible disconnected \(u\), \((\mathcal A_{\rm CFT}(u),\Omega_{\rm BD})\) is standard, so its Tomita modular flow exists. Geometric modular action, however, requires additional model-dependent structure.
\end{remark}

\begin{lemma}[Statistical independence of disjoint arcs]
\label{lem:statistical-independence}
If \(J_1,J_2\subsetneq S^1\) are disjoint open arcs, then \(\mathcal A_{\rm CFT}(J_1)\cap\mathcal A_{\rm CFT}(J_2)=\mathbb C\id\).
\end{lemma}
\begin{proof}
Since \(J_1\) and \(J_2\) are disjoint arcs, \(J_1\subseteq J_2^\perp\), so isotony and single-arc Haag duality give \(\mathcal A_{\rm CFT}(J_1)\subseteq\mathcal A_{\rm CFT}(J_2^\perp) =\mathcal A_{\rm CFT}(J_2)'\). Hence
\begin{align}
  \mathcal A_{\rm CFT}(J_1)\cap\mathcal A_{\rm CFT}(J_2)
  \subseteq
  \mathcal A_{\rm CFT}(J_2)'\cap\mathcal A_{\rm CFT}(J_2)
  = Z\bigl(\mathcal A_{\rm CFT}(J_2)\bigr)=\mathbb C\id ,
  \label{eq:statistical-independence-chain}
\end{align}
the last equality because \(\mathcal A_{\rm CFT}(J_2)\) is a factor.
\end{proof}

\begin{corollary}[Modular non-stationarity of local excitations]
\label{cor:modular-nonstationarity}
Let \(u\subsetneq S^1\) be a single open arc and \(J\) an open arc with \(\overline J\subset u\). For any non-scalar unitary \(V\in\mathcal A_{\rm CFT}(J)\) there is \(s\in\mathbb R\) with
\begin{align}
  e^{-isK_u}V\ket{\rm BD}\ \ne\ V\ket{\rm BD}.
  \label{eq:modular-nonstationarity}
\end{align}
\end{corollary}
\begin{proof}
Suppose \(e^{-isK_u}V\ket{\rm BD}=V\ket{\rm BD}\) for all \(s\). Since \(e^{isK_u}\ket{\rm BD}=\ket{\rm BD}\), this says \((\sigma_s^u(V)-V)\ket{\rm BD}=0\), and separatingness gives \(\sigma_s^u(V)=V\). Map \(u\) onto \((0,\infty)\) as in \Cref{lem:geometric-modular-flow}; the image of \(J\) is \((a_0,b_0)\) with \(0<a_0<b_0<\infty\). For \(s>(2\pi)^{-1}\log(b_0/a_0)\) the arc \(e^{2\pi s}(a_0,b_0)\) is disjoint from \((a_0,b_0)\), so \(\sigma^u_s(\mathcal A_{\rm CFT}(J))\) and \(\mathcal A_{\rm CFT}(J)\) are the algebras of disjoint arcs. Then \(V\) lies in their intersection, which is \(\mathbb C\id\) by \Cref{lem:statistical-independence}, contradicting that \(V\) is non-scalar.
\end{proof}

\begin{remark}
\label{rem:modular-nonstationarity-use}
Compact containment \(\overline J\subset u\) is essential: a dilation cannot move an arc sharing an endpoint with \(u\) off itself.
\end{remark}

\subsection{Stress-Tensor Realization and Comparison with Dilaton Gravity}
\label{sec:stress-tensor}

The comparison with JT gravity requires a full, non-chiral stress-tensor realization in addition to the abstract net.

\begin{assumption}[Stress-tensor realization]
\label{ass:stress-tensor}
For this subsection, \Cref{sec:discussion}, and \Cref{app:adm}, assume:
\begin{enumerate}[label=(\alph*)]
\item the matter sector is a full two-dimensional, non-chiral, unitary CFT with local net \(\mathcal A_{\rm full}\) on a Hilbert space \(\mathcal H_{\rm full}\), with Bunch--Davies vacuum \(\Omega_{\rm BD}^{\rm full}\), admitting an anomaly-free coupling to the metric, so in particular the left- and right-moving gravitational anomalies cancel, which for a unitary CFT requires \(c_L=c_R=:c\) \cite{AlvarezGaumeWitten1984};
\item this theory possesses a stress-energy tensor \(T_{\mu\nu}:=-2(\sqrt{-g})^{-1}\delta I_{\rm CFT}/\delta g^{\mu\nu}\) whose smearings against smooth test functions on \(\Sigma_0\) are densely defined closable operators with a common invariant dense domain containing \(\Omega_{\rm BD}^{\rm full}\);
\item \(H_{\rm stat}^{\rm full}\) is the global Noether charge of \eqref{eq:killing-field} in that full-CFT realization.
\end{enumerate}
This is a separate physical realization. In particular, \(H_{\rm stat}^{\rm full}\) and \(H_{\rm stat}^{\rm ch}\) act in different models.
\end{assumption}

Under \Cref{ass:stress-tensor}(c) the global generator is the full-circle integral
\begin{align}
  H_{\rm stat}^{\rm full}
  =\int_{S^1}\ed\chi\,\cos\chi\,T_{\eta\eta}(0,\chi),
  \label{eq:hstat-stress-tensor}
\end{align}
whereas the object that appears in the canonical JT constraint of \Cref{app:adm} is the \emph{one-sided matter charge}
\begin{align}
  H_{\rm mat}^{R}
  :=\int_{-\pi/2}^{\pi/2}\ed\chi\,\cos\chi\,T_{\eta\eta}(0,\chi).
  \label{eq:hmat-definition}
\end{align}

\begin{remark}[Global and one-sided charges]
\label{rem:hstat-neq-hmat}
Formally
\begin{align}
  H_{\rm stat}^{\rm full}=H_{\rm mat}^{R}+H_{\rm mat}^{R^\perp},
  \qquad
  H_{\rm mat}^{R^\perp}
  =\int_{\pi/2}^{3\pi/2}\ed\chi\,\cos\chi\,T_{\eta\eta}(0,\chi),
  \label{eq:two-sided-split}
\end{align}
and \(\cos\chi<0\) on the interior of \(R^\perp\). Hence the global generator is a two-sided charge, whereas the canonical JT constraint below contains the one-sided charge \(H_{\rm mat}^{R}\). Equality of their derivations on \(\mathcal A_{\rm full}(R)\) would require a separate domain and self-adjointness analysis at the bifurcation points.
\end{remark}

Subject to \Cref{ass:stress-tensor}, the gravitational background is JT gravity \cite{Teitelboim,Jackiw} with positive cosmological constant coupled to the matter CFT,
\begin{align}
  I = \frac{1}{16\pi G_2}\int_M \ed^2x\,\sqrt{-g}\,
      \Phi\Bigl(R-\frac{2}{L^2}\Bigr)
 +\frac{1}{8\pi G_2}\int_{\partial M}
      \epsilon\,\sqrt{|\gamma|}\,\Phi K
 +I_{\rm corner}+I_{\rm CFT}[g] ,
  \label{eq:jt-action}
\end{align}
with \(\epsilon=n^\mu n_\mu=\pm1\); null segments require their corresponding null-boundary term and joints require \(I_{\rm corner}\). Both must be specified whenever a static patch is treated as a manifold with boundary. In the positive-curvature conventions of \cite{MTY,CotlerJensenMaloney} the metric equation of motion \(R=2/L^2\) is solved by \eqref{eq:global-ds-metric} and the dilaton equation is
\begin{align}
  (g_{\mu\nu}\Box-\nabla_\mu\nabla_\nu)\Phi
  +\frac{1}{L^2}g_{\mu\nu}\Phi
  =8\pi G_2\,T_{\mu\nu} .
  \label{eq:dilaton-eom}
\end{align}
The ADM reduction is carried out in \Cref{app:adm}. On the time-reflection-symmetric slice the lapse constraint reduces to
\begin{align}
  \frac{1}{8\pi G_2}\bigl(\partial_\chi^2\Phi+\Phi\bigr)
  +T_{\eta\eta}=0 ,
  \label{eq:adm-constraint-density}
\end{align}
where a common factor \(1/L\) has been removed; the \(\eta\eta\) component of \eqref{eq:dilaton-eom} evaluated at \(\eta=0\) reproduces \eqref{eq:adm-constraint-density} exactly, which fixes the relative sign in \eqref{eq:dilaton-eom}. Smearing \eqref{eq:adm-constraint-density} with the Killing lapse \(N_\xi=L\cos\chi\) over \(R\) gives, as computed in \Cref{app:adm-smearing},
\begin{align}
  \frac{\Phi_++\Phi_-}{8\pi G_2}+H_{\rm mat}^{R}=0,
  \qquad \Phi_\pm:=\Phi(0,\pm\pi/2).
  \label{eq:boundary-term}
\end{align}
The points \((0,\pm\pi/2)\) are the two spatial endpoints, or bifurcation points, of the patch on \(\Sigma_0\); they are not its future and past tips, which lie at \((\pm\pi/2,0)\) in the conformal completion.

\begin{proposition}[Classical constraint-surface charge relation]
\label{prop:classical-charge-relation}
Work in classical JT gravity coupled to a matter theory satisfying \Cref{ass:stress-tensor}(a). Let \(Q_{\rm JT}^{R}:=-(\Phi_++\Phi_-)/(8\pi G_2)\) be the endpoint improvement term generated by allowing the endpoint values of \(\Phi\) to vary in the smeared constraint, as derived in \Cref{app:adm-smearing}. Then, on the constraint surface of the time-reflection-symmetric slice,
\begin{align}
  Q_{\rm JT}^{R}=H_{\rm mat}^{R}.
  \label{eq:classical-charge-relation}
\end{align}
\end{proposition}

\begin{proof}
Immediate from \eqref{eq:boundary-term} and the definition of \(Q_{\rm JT}^R\); the derivation, including the orientation, the boundary terms and the \(L\)-independence of the cancellation, is \Cref{app:adm-smearing}.
\end{proof}

\begin{remark}[Classical character of the charge relation]
\label{rem:classical-charge-status}
Equation~\eqref{eq:classical-charge-relation} is a relation between classical phase-space functions on the constraint surface. A quantum operator relation would additionally require a differentiable generator with specified null-boundary and corner conditions, together with a quantization of the endpoint data \(\Phi_\pm\).
\end{remark}

\begin{remark}[Comparison of normalizations]
\label{rem:normalization-compatibility}
The chiral-net identity \(K_R=2\pi H_{\rm stat}^{\rm ch}\) and the full-CFT/JT relation \(Q^R_{\rm JT}=H^R_{\rm mat}\) use the same classical Killing field \eqref{eq:killing-field} and the same dimensionless parameter \(\tau\). Their normalization conventions are therefore geometrically comparable. Since the first identity belongs to the chiral model and the second is classical and full-CFT based, they do not furnish an operator relation between \(K_R\) and
\(Q^R_{\rm JT}\).
\end{remark}

\subsection{The Reference Crossed Product and Its Finite Corner}
\label{sec:reference-crossed-product}

From here to the end of \Cref{sec:activation-model} only \Cref{ass:regional-net} is used.

For the algebraic construction we use the regular representation of the continuous core. The auxiliary Hilbert space is \(L^2(\mathbb R,\ed q)\), where \(q_{\rm clk}\) is a dimensionless modular-clock coordinate. To keep the crossed-product formulas readable we write \(q:=q_{\rm clk}\) and \(h:=p_{\rm clk}\), where
\begin{align}
 p_{\rm clk}:=-i\partial_{q_{\rm clk}},\qquad
 [q_{\rm clk},p_{\rm clk}]=i,\qquad V(s):=e^{is p_{\rm clk}} .
 \label{eq:clock-variables}
\end{align}
The symbol \(h\) always denotes this auxiliary generator, never the ADM spatial metric of \Cref{app:adm}. For \(x\in\mathcal A_{\rm CFT}(u)\) put
\begin{align}
 D_u(x):=\sigma_q^u(x)
 =e^{-iqK_u}(x\otimes\id)e^{iqK_u},
 \label{eq:regular-core-representation}
\end{align}
acting by fibrewise multiplication, \((D_u(x)\xi)(q)=\sigma_q^u(x)\xi(q)\), while \((V(s)\xi)(q)=\xi(q+s)\). Hence
\begin{align}
 V(s)D_u(x)V(s)^*=D_u\bigl(\sigma_s^u(x)\bigr),
 \label{eq:core-covariance}
\end{align}
which is the crossed-product covariance relation. This is the observer-clock mechanism of \cite{CLPW}, expressed in modular units. For \(u=R\), \(q=-\tau/(2\pi)\) along static evolution by \eqref{eq:static-modular-relation}. The variables \(q\) and \(h\) remain auxiliary in the JT comparison of \Cref{app:adm}.

The whole-circle algebra \(\mathcal A_{\rm CFT}(S^1)=B(\mathcal H_{\rm CFT})\) is a type-\(\mathrm I_\infty\) factor with inner modular flow. We therefore use the type-\(\mathrm{III}_1\) local factor \(\mathcal A_{\rm CFT}(R)\) and set
\begin{align}
  \hat{\mathcal{M}}
  :=\bigl(\mathcal A_{\rm CFT}(R)\rtimes_{\sigma^R}\mathbb R\bigr)''
  =\{D_R(x),V(s):x\in\mathcal A_{\rm CFT}(R),\ s\in\mathbb R\}'' ,
  \label{eq:reference-core}
\end{align}
which is a type-\(\mathrm{II}_\infty\) factor by Takesaki's continuous-core theorem \cite{Takesaki}. The dual action is implemented by \(Z(r):=e^{irq}\). Since \(Z(r)^*hZ(r)=h+r\), it obeys
\begin{align}
 \widehat\sigma_r(D_R(x))=D_R(x),\qquad
 \widehat\sigma_r(V(s))=e^{-irs}V(s).
 \label{eq:dual-action-convention}
\end{align}
Let \(\widehat\Tr\) be the semifinite trace representative selected by the normalized vacuum state and the Haar/Fourier convention fixed in the proof of \Cref{prop:finite-trace}; it obeys
\begin{align}
 \widehat\Tr\circ\widehat\sigma_r=e^{-r}\widehat\Tr .
 \label{eq:trace-scaling}
\end{align}
Following the positive-spectrum corner used in \cite{CLPW}, put \(P:=\mathbf1_{[0,\infty)}(h)\) and define
\begin{align}
  \mathcal{M}_{\rm ref} := P\,\hat{\mathcal{M}}\,P .
  \label{eq:global-algebra-def}
\end{align}

\begin{proposition}[Finite trace on the positive-clock-spectrum corner]
\label{prop:finite-trace}
The corner \(\mathcal M_{\rm ref}=P\hat{\mathcal M}P\) carries a finite, normalized trace. Consequently \(\mathcal M_{\rm ref}\) is a type-\(\mathrm{II}_1\) factor.
\end{proposition}

\begin{proof}
The continuous core of a type-\(\mathrm{III}_1\) factor is a type-\(\mathrm{II}_\infty\) factor and carries a faithful normal semifinite trace, unique up to a positive scalar \cite{Takesaki}. Two further steps are needed: membership of \(P\), and the value of its trace.

\emph{Membership.} Every \(V(s)\) belongs to \(\hat{\mathcal M}\) by \eqref{eq:reference-core}. Hence its Stone generator \(h\) is affiliated with \(\hat{\mathcal M}\), and Borel functional calculus gives \(P=\mathbf1_{[0,\infty)}(h)\in\hat{\mathcal M}\). This uses the full unitary group \(V(s)\), not merely the fact that \(P\) acts on the same Hilbert space.

\emph{Normalization.} Let \(\varphi_R(x)=\langle\Omega_{\rm BD}, x\,\Omega_{\rm BD}\rangle\), let \(\widehat\varphi_R\) be its dual weight on the crossed product, and let \(\rho_R\) be the positive affiliated operator with \(\rho_R^{is}=V(s)\). In the regular representation \(V(s)=e^{ish}\), so uniqueness of the Stone generator gives \(\log\rho_R=h\). For \(x\in\hat{\mathcal M}_+\) the trace selected by the dual-weight construction
is
\begin{align}
 \widehat\Tr(x)
 =\widehat\varphi_R\!\left(
    \rho_R^{-1/2}x\,\rho_R^{-1/2}\right).
 \label{eq:dual-weight-trace}
\end{align}
Because \(\rho_R\) is generally unbounded, the right-hand side is not an ordinary product of bounded operators: it denotes the value of \(\widehat\varphi_R\) on the associated positive quadratic form, equivalently on the corresponding element of the extended positive cone, obtained by the standard spatial-derivative calculus for weights, and is allowed to be \(+\infty\). Writing \(\widehat\varphi_R(\rho_R^{-1}x)\) is therefore only shorthand, literal when the operators commute and the product lies in the domain of the weight.

Fix the normalization explicitly. The unitary Fourier transform on the clock factor is
\[
 (\mathcal F\xi)(x)
  =(2\pi)^{-1/2}\int_{\mathbb R}e^{-iqx}\xi(q)\,\ed q ,
\]
so that \(h=-i\partial_q\) becomes multiplication by \(x\) and \(V(s)\) becomes multiplication by \(e^{isx}\). For integrated crossed-product elements we use the Haar measure \(\ed\mu(s)=\ed s/(2\pi)\):
\[
 V(g):=\int_{\mathbb R}g(s)V(s)\,\ed\mu(s),\qquad
 \widetilde g(x):=\frac1{2\pi}\int_{\mathbb R}
      e^{isx}g(s)\,\ed s ,
 \qquad g(s)=\int_{\mathbb R}e^{-isx}\widetilde g(x)\,\ed x ,
\]
whence \(g(0)=\int_{\mathbb R}\widetilde g(x)\,\ed x\). This convention fixes the trace representative and prevents an implicit factor of \(2\pi\).

If \(f\ge0\) is Borel with \(\int_{\mathbb R}e^{-x}f(x)\,\ed x<\infty\), the canonical operator-valued weight from the crossed product onto \(\mathcal A_{\rm CFT}(R)\) gives
\begin{align}
 \widehat\Tr\!\left(f(h)\right)
 =\widehat\varphi_R\!\left(\rho_R^{-1/2}
       f(\log\rho_R)\rho_R^{-1/2}\right)
   =\varphi_R(\id)\int_{\mathbb R}e^{-x}f(x)\,\ed x .
 \label{eq:core-spectral-trace}
\end{align}
Here \(f(h)=f(\log\rho_R)\) commutes with \(\rho_R=e^{h}\), so the symmetric extended-positive-cone expression reduces to the spectral function \(e^{-h}f(h)\); this simplification is used only for Borel functions of \(h\). The result is the Fourier/Haar computation in the regular crossed product \cite{Takesaki,CLPW}, and it is strictly stronger than \eqref{eq:trace-scaling}, which alone fixes the spectral density only up to a constant. Taking monotone approximations to \(f=\mathbf1_{[0,\infty)}\) and using \(\varphi_R(\id)=1\),
\begin{align}
  \widehat\Tr(P)=\int_0^\infty e^{-x}\,\ed x=1 .
  \label{eq:trace-of-identity}
\end{align}
Thus \(P\) is a nonzero finite projection. A nonzero corner \(P\mathcal NP\) of a factor \(\mathcal N\) with \(P\in\mathcal N\) is again a factor, and a finite corner of a type-\(\mathrm{II}_\infty\) factor is a type-\(\mathrm{II}_1\) factor. Its normalized trace is
\begin{align}
 \tau_{R,P}(PxP)
 :=\frac{\widehat\Tr(PxP)}{\widehat\Tr(P)}
 =\widehat\Tr(PxP).
 \label{eq:normalized-corner-trace}
\end{align}
\end{proof}

\begin{remark}[Normalization of the trace]
\label{rem:trace-scale}
As an abstract type-\(\mathrm{II}_\infty\) factor, \(\hat{\mathcal M}\) has a faithful normal semifinite trace only up to a positive scalar. Equation~\eqref{eq:trace-of-identity} is nevertheless derived, not postulated, once one has chosen the normalized vacuum state, the Haar/Fourier convention above, and the dual-weight trace \eqref{eq:dual-weight-trace}. Rescaling the semifinite trace multiplies both sides of \(\widehat\Tr(P)=1\) by a common positive factor. This is the same ambiguity that leaves entropies on a type-\(\mathrm{II}_\infty\) algebra defined only up to an additive shift \cite{PeningtonWitten2023}. The restriction \(x\mapsto\widehat\Tr(PxP)\) is an unnormalized finite trace until one divides by \(\widehat\Tr(P)\); the last equality in \eqref{eq:normalized-corner-trace} uses this particular dual-weight normalization.
\end{remark}

\begin{remark}[Algebra types]
\label{rem:algebra-types}
Collecting the types: \(\mathcal A_{\rm CFT}(u)\) is type \(\mathrm{III}_1\) for every admissible \(u\); \(\hat{\mathcal M}(u)\) is type \(\mathrm{II}_\infty\); \(P\) is a projection of finite nonzero trace in \(\hat{\mathcal M}(u)\); and the corner \(P\hat{\mathcal M}(u)P\) is type \(\mathrm{II}_1\). Thus the static-patch matter algebra \(\mathcal A_{\rm CFT}(R)\) is type \(\mathrm{III}_1\), while its finite continuous-core corner \(\mathcal M_{\rm ref}\) is type \(\mathrm{II}_1\).
\end{remark}

The concrete assignment of \Cref{sec:subregion-algebra-def} depends on the relative position of the reference and protected-gap embeddings. It is covariant under a \emph{simultaneous} spatial rotation: if \(\rho\) is a rotation and \(U_\rho\) its vacuum implementer (acting trivially on the clock), and the superscript records the chosen reference arc, then
\begin{align}
 \mathcal A_{\rm model}^{\,\rho R}(\rho a)
 =U_\rho\,\mathcal A_{\rm model}^{\,R}(a)\,U_\rho^* .
 \label{eq:reference-covariance}
\end{align}
The superscript is suppressed after fixing \(R\) as in \eqref{eq:reference-region}.

\begin{remark}[Reference algebra]
\label{rem:what-mglobal-represents}
\(\mathcal M_{\rm ref}\) is the concrete corner algebra for the fixed arc \(R\), chosen from a unitarily equivalent family of single-patch corners. A construction containing both opposite static patches would instead use the two-observer framework of \cite{CLPW}.
\end{remark}

\subsection{Ambient Representation of Dressed Regional Algebras}
\label{sec:ambient-representation}

To form \(\mathcal M(\widetilde A)'\cap\mathcal M_{\rm ref}\), the regional crossed products must be represented in a common \(B(\mathcal H)\). The next lemma supplies this representation.

\begin{lemma}[Ambient representation of dressed subregion algebras]
\label{lem:ambient-representation}
Let \(u\subsetneq S^1\) be admissible, with both \(u\) and \(u^\perp\) containing a nonempty open arc. Then:
\begin{enumerate}[label=(\alph*)]
\item \(\ket{\rm BD}\) is cyclic and separating for \(\mathcal A_{\rm CFT}(u)\), so Tomita--Takesaki theory supplies \(\Delta_u\) and \(K_u=-\log\Delta_u\) with \(K_u\ket{\rm BD}=0\) and \(\sigma_t^u=\operatorname{Ad}(e^{-itK_u})\) preserving \(\mathcal A_{\rm CFT}(u)\).
\item Dressing \(\mathcal A_{\rm CFT}(u)\) by \(K_u\) on the \emph{same} clock space \(L^2(\mathbb R)\) used for \(R\), and applying the \emph{same} projection \(P=\mathbf1_{[0,\infty)}(h)\), defines
\begin{align}
  \hat{\mathcal M}(u)
  &:=\bigl(\mathcal A_{\rm CFT}(u)\rtimes_{\sigma^u}\mathbb R\bigr)'',
  &
  \mathcal M(D(u))
  &:=P\hat{\mathcal M}(u)P
  \subset B\bigl(\mathcal H_{\rm CFT}\otimes L^2(\mathbb R)\bigr) .
  \label{eq:ambient-dressed-algebra}
\end{align}
\item \(\hat{\mathcal M}(u)\) is the type-\(\mathrm{II}_\infty\) continuous core of \(\mathcal A_{\rm CFT}(u)\), its dual-weight trace \(\widehat\Tr_u\) satisfies
\begin{align}
 0<\widehat\Tr_u(P)=1<\infty,\qquad
 \tau_{u,P}(x):=\frac{\widehat\Tr_u(x)}{\widehat\Tr_u(P)}
 \quad \bigl(x\in P\hat{\mathcal M}(u)P\bigr),
 \label{eq:regional-corner-trace}
\end{align}
and \(\mathcal M(D(u))\) is a type-\(\mathrm{II}_1\) factor. For \(u=R\) the construction reduces to \(\mathcal M_{\rm ref}\) of \eqref{eq:global-algebra-def}.
\end{enumerate}
\end{lemma}
\begin{proof}
The discussion after \Cref{ass:regional-net} shows that \(\mathcal A_{\rm CFT}(u)\) is a type-\(\mathrm{III}_1\) factor and that the vacuum is cyclic and separating for it. Tomita--Takesaki theory gives (a), and the regular representation \eqref{eq:regular-core-representation} gives (b). Takesaki's continuous-core theorem makes \(\hat{\mathcal M}(u)\) a type-\(\mathrm{II}_\infty\) factor \cite{Takesaki}.

The membership and trace argument is uniform in \(u\). The defining generators include the same full unitary group \(V(s)=e^{ish}\), so \(h\) is affiliated with every represented core and the common spectral projection \(P=\mathbf1_{[0,\infty)}(h)\) belongs to each one. Let \(\varphi_u\) be the normalized vacuum restriction and \(\widehat\varphi_u\) its dual weight; the corresponding density \(\rho_u\) obeys \(\rho_u^{is}=V(s)\), hence \(\log\rho_u=h\). Repeating \eqref{eq:dual-weight-trace}--\eqref{eq:core-spectral-trace} verbatim gives
\[
 \widehat\Tr_u(f(h))
 =\int_{\mathbb R}e^{-x}f(x)\,\ed x,\qquad
 \widehat\Tr_u(P)=1 ,
\]
so \(P\) is a nonzero finite projection in every regional core. As in \Cref{prop:finite-trace}, a nonzero corner of a factor by a projection of the factor is a factor, and a finite corner of a type-\(\mathrm{II}_\infty\) factor is type \(\mathrm{II}_1\). This proves (c). No physical inverse temperature is inserted for a general region: Tomita modular time is already normalized so that \(\varphi_u\) is \(1\)-KMS.
\end{proof}

Because \(\mathcal H_{\rm CFT}\) and \(L^2(\mathbb R)\) are fixed, \Cref{lem:ambient-representation} places all regional cores in the same \(B(\mathcal H_{\rm CFT}\otimes L^2(\mathbb R))\). Commutants and intersections taken below are therefore well defined even though the regional algebras are not nested.

\begin{remark}[Auxiliary clock]
\label{rem:algebraic-vs-geometric-clock}
\Cref{lem:ambient-representation} uses the modular Hamiltonian \(K_u\) of each region together with one auxiliary \(L^2(\mathbb R)\). The positive-spectrum condition concerns its generator \(h\), while a gravitational implementation would require a clock derived from the JT phase space.
\end{remark}

\begin{remark}[Regional traces]
\label{rem:common-projection-traces}
The fixed regular representation supplies one operator \(P=\mathbf1_{[0,\infty)}(h)\) for every region, because every represented core contains the same \(V(s)\) and hence every spectral projection of its generator. Each regional core instead carries its own dual-weight trace \(\widehat\Tr_u\), normalized by the regional vacuum state and the common Haar/Fourier convention so that \(\widehat\Tr_u(P)=1\). These traces are regional objects; the ambient representation carries no common semifinite trace restricting to all of them.
\end{remark}

\subsection{Algebras of Protected Diamonds}
\label{sec:diamond-algebras}

A protected gap \(b\) with \(\beta=|b|\ge\pi\) has a domain of dependence \(B:=D(b)\) meeting both conformal boundaries. Using the same auxiliary modular clock and the same projection \(P\) we set
\begin{align}
  \mathcal{M}(B) := P\bigl( \mathcal{A}_{\rm CFT}(b)
  \rtimes_{\sigma^{b}} \mathbb{R} \bigr)''P ,
  \label{eq:diamond-algebra}
\end{align}
where \(\sigma^b\) is the group \eqref{eq:modular-convention} for the restricted Bunch--Davies state. \Cref{lem:ambient-representation} realizes \(\mathcal M(B)\) on the common ambient Hilbert space, and \Cref{lem:conformal-covariance} gives \(\mathcal M(B)=U_g\mathcal M_{\rm ref}U_g^{-1}\) for a M\"obius map \(g\) with \(g(R)=b\), the unitary \(U_g\) acting trivially on the clock factor. By \Cref{thm:gap-protection} two disjoint gaps cannot both be protected, so no multi-gap factorization is needed in this geometry. We also adopt the convention \(\mathcal M(\varnothing):=\mathbb C\id\).

A M\"obius transformation cannot map a connected arc onto a disconnected region, so for disconnected regions only the common ambient representation of \Cref{lem:ambient-representation} is available; there is no conjugation argument to fall back on.

\section{The Inclusion Criterion and Its Obstruction}
\label{sec:isotony}

\begin{remark}[KLM and regional inclusions]
\label{rem:klm-scope}
Ordinary isotony gives \(\mathcal A_{\rm CFT}(u)\subseteq\mathcal A_{\rm CFT}(v)\) for \(u\subseteq v\). For a separated \(n\)-arc region \(E\), the Kawahigashi--Longo--M\"uger (KLM) inclusion is instead
\begin{align}
  \mathcal A_{\rm CFT}(E)\subseteq
  \widehat{\mathcal A}_{\rm CFT}(E)
  :=\mathcal A_{\rm CFT}(E^\perp)' ,
  \label{eq:klm-inclusion}
\end{align}
which is trivial for a single arc under Haag duality and nontrivial for disconnected \(E\). KLM analyse \eqref{eq:klm-inclusion} for a local irreducible chiral net with the split property and Haag duality \cite{KLM2001}. When the net is completely rational, the two-arc \(\mu\)-index \(\mu_{\mathcal A}=[\widehat{\mathcal A}_{\rm CFT}(E): \mathcal A_{\rm CFT}(E)]\) is finite, independent of the two-arc region, and equals the global dimension of the sector category; the \(n\)-arc inclusions are treated separately in their Section~6 and their index is not a region-independent copy of the two-arc index. Its canonical expectation \(\widehat{\mathcal A}_{\rm CFT}(E)\to\mathcal A_{\rm CFT}(E)\) differs from the expectation \(\mathcal A_{\rm CFT}(v)\to\mathcal A_{\rm CFT}(u)\) governing an arbitrary nested pair below. The KLM theorem therefore provides context for multi-arc duality, while the regional continuous-core embeddings require separate compatibility data.
\end{remark}

We retain \Cref{ass:regional-net} throughout this section. The same arguments apply to a full two-dimensional net once regional standardness, factoriality, and covariance are verified directly.

Every region is represented on the same \(\mathcal H_{\rm CFT}\otimes L^2(\mathbb R,\ed q)\), with the same multiplication operator \(q\), the same translations \(V(s)=e^{ish}\), the same dual implementers \(Z(r)=e^{irq}\), and the same projection \(P=\mathbf1_{[0,\infty)}(h)\). Only the fibre representation \(D_u(x)\xi(q)=\sigma_q^u(x)\xi(q)\) depends on \(u\). It is faithful because \(\lVert D_u(x)\rVert=\operatorname*{ess\,sup}_q \lVert\sigma_q^u(x)\rVert=\lVert x\rVert\). Together with \(V\) this is the faithful regular representation of the von Neumann crossed product \cite{Takesaki}, so \(\hat{\mathcal M}(u)\) always denotes the concrete represented image, not an abstract crossed product up to isomorphism.

\begin{theorem}[Shared-clock inclusion criterion]
\label{thm:net-compatibility-criterion}
Let \(u\subset v\subsetneq S^1\) be admissible regions in the class of \Cref{lem:ambient-representation}. The following are equivalent.
\begin{enumerate}[label=(\roman*)]
\item The concrete regular images satisfy
\(\hat{\mathcal M}(u)\subseteq\hat{\mathcal M}(v)\).
\item Their common spectral corners satisfy
\(P\hat{\mathcal M}(u)P\subseteq P\hat{\mathcal M}(v)P\).
\item \(\sigma^v_t(x)=\sigma^u_t(x)\) for every \(x\in\mathcal A_{\rm CFT}(u)\) and every \(t\in\mathbb R\); that is, \(v\)'s own modular flow restricts on \(\mathcal A_{\rm CFT}(u)\) to \(u\)'s own modular flow.
\item \(\mathcal A_{\rm CFT}(u)\) is globally invariant under \(\sigma^v\).
\item There is a normal conditional expectation \(E:\mathcal A_{\rm CFT}(v)\to\mathcal A_{\rm CFT}(u)\) preserving the vacuum, \(\omega_{\rm BD}\circ E=\omega_{\rm BD}\).
\end{enumerate}
The inclusions in (i) and (ii) are literal inclusions in the displayed common representation, and no trace comparison enters the criterion. The statement applies equally to connected and disconnected regions; the only regional inputs are \(\mathcal A_{\rm CFT}(u)\subseteq\mathcal A_{\rm CFT}(v)\) and faithful normal modular data for both regions. Conditions (iv) and (v) mention neither the clock nor the projection: whether the dressed cores are nested is decided by a property of the undressed inclusion \(\mathcal A_{\rm CFT}(u)\subseteq\mathcal A_{\rm CFT}(v)\) alone.
\end{theorem}
\begin{proof}
\emph{(iii)\(\Rightarrow\)(i).} For \(x\in\mathcal A_{\rm CFT}(u)\), condition (iii) gives the operator identity \(D_u(x)=D_v(x)\). Both represented cores contain the identical translation group \(V(s)\). Hence every generator of \(\hat{\mathcal M}(u)\) lies in \(\hat{\mathcal M}(v)\).

\emph{(i)\(\Rightarrow\)(ii).} The common projection \(P\) belongs to both cores by \Cref{lem:ambient-representation}; multiplying the inclusion in (i) on both sides by \(P\) gives (ii).

\emph{(i)\(\Rightarrow\)(iii).} The dual action \(\widehat\sigma_r\) of \eqref{eq:dual-action-convention} is implemented on the ambient space by the single unitary group \(Z(r)=e^{irq}\), which commutes with every fibrewise multiplication operator; it therefore preserves both cores and acts on each as its own dual action. Takesaki's fixed-point theorem \cite{Takesaki} identifies the fixed algebra of \(\widehat\sigma\) in \(\hat{\mathcal M}(w)\) with \(D_w(\mathcal A_{\rm CFT}(w))\), for \(w\in\{u,v\}\). Thus (i) implies that for each \(x\in\mathcal A_{\rm CFT}(u)\) there is \(y\in\mathcal A_{\rm CFT}(v)\) with \(D_u(x)=D_v(y)\). Both sides are decomposable operators on \(L^2(\mathbb R;\mathcal H_{\rm CFT})\), and a decomposable operator determines its operator field almost everywhere, so \(\sigma_q^u(x)=\sigma_q^v(y)\) for almost every \(q\). Both fields are \(\sigma\)-weakly continuous in \(q\), so all their matrix elements are continuous and equality almost everywhere upgrades to equality for every \(q\). At \(q=0\) this reads \(x=y\), and at general \(q=t\) it is exactly (iii). Faithfulness of the regular representations excludes any kernel qualification.

\emph{(iii)\(\Leftrightarrow\)(iv)\(\Leftrightarrow\)(v).} We use Takesaki's theorem \cite{Takesaki1972} in the following form. Let \(\mathcal N\subseteq \mathcal M\) be von Neumann algebras and \(\omega\) a faithful normal state on \(\mathcal M\). Then \(\mathcal N\) is globally invariant under \(\sigma^\omega\) if and only if there is a normal conditional expectation \(\mathcal M\to\mathcal N\) preserving \(\omega\), in which case that expectation is unique and
\(\sigma^\omega_t|_{\mathcal N}=\sigma^{\omega|_{\mathcal N}}_t\).

Here \(\mathcal M=\mathcal A_{\rm CFT}(v)\), \(\mathcal N=\mathcal A_{\rm CFT}(u)\), and \(\omega\) is the vacuum, which restricts to a faithful normal state on \(\mathcal A_{\rm CFT}(v)\) because \(\Omega_{\rm BD}\) is separating for it. The stated equivalence of invariance with the existence of the expectation is (iv)\(\Leftrightarrow\)(v). If (iii) holds then \(\sigma^v_t(x)=\sigma^u_t(x)\in\mathcal A_{\rm CFT}(u)\), giving (iv). Conversely, if (iv) holds then the last clause identifies \(\sigma^v_t|_{\mathcal A_{\rm CFT}(u)}\) with the modular group of \(\omega_{\rm BD}|_{\mathcal A_{\rm CFT}(u)}\); that modular group is \(\sigma^u\) by definition and modular groups are unique, so (iii) follows.

\emph{(ii)\(\Rightarrow\)(i).} From \([q,h]=i\) one gets \(Z(r)PZ(r)^*=\mathbf1_{[r,\infty)}(h)=:P_r\), so each \(P_r\) belongs to both represented cores and \(P_r\to\id\) strongly as \(r\to-\infty\). Since the same \(Z(r)\) implements the dual action on both algebras, conjugating (ii) gives
\[
 P_r\hat{\mathcal M}(u)P_r\subseteq
 P_r\hat{\mathcal M}(v)P_r\subseteq\hat{\mathcal M}(v)
\]
for every \(r\). For \(x\in\hat{\mathcal M}(u)\) we have \(P_rxP_r-x=P_rx(P_r-\id)+(P_r-\id)x\to0\) strongly, and strong closedness of \(\hat{\mathcal M}(v)\) gives \(x\in\hat{\mathcal M}(v)\). Thus compression by this common \(P\) neither weakens nor obscures the criterion. The argument uses only the common dual action and \(P_r\to\id\), not any regional trace.
\end{proof}

\begin{remark}[General faithful normal states]
\label{rem:criterion-general-states}
The proof extends to any faithful normal state \(\psi\) on \(\mathcal A_{\rm CFT}(v)\), realized in standard form on \(\mathcal H_{\rm CFT}\), and set \(\sigma^v:=\sigma^{\psi}\) and \(\sigma^u:=\sigma^{\psi|_{\mathcal A_{\rm CFT}(u)}}\). Then (i)--(v) remain equivalent with \(\omega_{\rm BD}\) replaced by \(\psi\) throughout, so that (v) reads: there is a \(\psi\)-preserving normal conditional expectation \(\mathcal A_{\rm CFT}(v)\to\mathcal A_{\rm CFT}(u)\). The two flows must come from one state on the larger algebra together with its restriction, which is exactly the situation in \Cref{sec:split-repair}, where \(\psi=\phi_S\) and
\(\psi|_{\mathcal N_T}=\phi_T\).
\end{remark}

\begin{remark}[Cocycle-modified embeddings]
\label{rem:cocycle-scope}
If two actions are related by an inner cocycle, one can sometimes construct an abstract crossed-product homomorphism carrying one implementing group to a cocycle-twisted version of the other. That is a different embedding. In the canonical shared-\(q\) realization above, both cores are required to contain the same \(V(s)\), while their fixed copies of the matter algebras are the specified \(D_u\) and \(D_v\). The fixed-point argument then forces \(\sigma_t^v|_{\mathcal A_{\rm CFT}(u)}=\sigma_t^u\); a cocycle or inner perturbation defines a different embedding. Coherent cocycle-modified embeddings for an entire regional family remain an open possibility.
\end{remark}

\begin{proposition}[Vacuum dressing obstructs isotony]
\label{prop:net-compatibility-fails}
Let \(u\subsetneq v\subsetneq S^1\) be admissible regions covered by \Cref{lem:ambient-representation}, and suppose that \(v\setminus\overline u\) contains a nonempty open arc. Then \(\sigma^v_t(\mathcal A_{\rm CFT}(u))\ne\mathcal A_{\rm CFT}(u)\) for some \(t\), and consequently
\[
 \hat{\mathcal M}(u)\not\subseteq\hat{\mathcal M}(v),
 \qquad
 P\hat{\mathcal M}(u)P\not\subseteq P\hat{\mathcal M}(v)P .
\]
Thus the common-clock family built from the regions' own vacuum modular flows is not an isotonic net. More explicitly, there are \(x\in\mathcal A_{\rm CFT}(u)\) and \(t\in\mathbb R\) with \(\sigma_t^v(x)\ne\sigma_t^u(x)\), and for any such \(x\),
\[
 D_u(x)\in\hat{\mathcal M}(u)\setminus\hat{\mathcal M}(v),
 \qquad
 PD_u(x)P\in
 P\hat{\mathcal M}(u)P\setminus P\hat{\mathcal M}(v)P .
\]
\end{proposition}
\begin{proof}
By \Cref{thm:net-compatibility-criterion} it is enough to refute condition (v), the existence of a vacuum-preserving normal conditional expectation
\[
 E:\mathcal A_{\rm CFT}(v)\longrightarrow\mathcal A_{\rm CFT}(u),
 \qquad \omega_{\rm BD}\circ E=\omega_{\rm BD} .
\]
Suppose such an \(E\) existed. Since \(v\) is open and \(\overline u\) closed, \(v\setminus\overline u\) is open, so it contains an open arc \(I\). As \(I\) and \(u\) are disjoint, locality gives \(\mathcal A_{\rm CFT}(I)\subseteq\mathcal A_{\rm CFT}(u)'\). For \(y\in\mathcal A_{\rm CFT}(I)\) and \(x\in\mathcal A_{\rm CFT}(u)\), bimodularity of \(E\) and \(xy=yx\) give
\[
 xE(y)=E(xy)=E(yx)=E(y)x ,
\]
so \(E(y)\in\mathcal A_{\rm CFT}(u)\cap\mathcal A_{\rm CFT}(u)'\), which is \(\mathbb C\id\) because \(\mathcal A_{\rm CFT}(u)\) is a factor. Applying \(\omega_{\rm BD}\) fixes the scalar: \(E(y)=\omega_{\rm BD}(y)\id\). Hence
\[
 \omega_{\rm BD}(xy)=\omega_{\rm BD}(E(xy))
 =\omega_{\rm BD}(x)\,\omega_{\rm BD}(y)
 \qquad\bigl(x\in\mathcal A_{\rm CFT}(u)\bigr).
\]
Because \(\mathcal A_{\rm CFT}(I)\ne\mathbb C\id\) there is a nonzero \(y\in\mathcal A_{\rm CFT}(I)\) with \(\omega_{\rm BD}(y)=0\); the displayed identity then says \(\langle x^*\Omega_{\rm BD},y\Omega_{\rm BD}\rangle=0\) for every \(x\in\mathcal A_{\rm CFT}(u)\). Cyclicity for \(\mathcal A_{\rm CFT}(u)\) gives \(y\Omega_{\rm BD}=0\), and separatingness for \(\mathcal A_{\rm CFT}(I)\) gives \(y=0\), a contradiction.

So (v) fails, and with it (i)--(iv). In particular \(\mathcal A_{\rm CFT}(u)\) is not globally \(\sigma^v\)-invariant, and both displayed inclusions fail.

The inclusion \(\sigma^v_t(\mathcal A_{\rm CFT}(u))\subseteq\mathcal A_{\rm CFT}(u)\) cannot hold for every \(t\): applying it at \(-t\) would otherwise give equality for every \(t\). Hence there are \(t\) and \(x\in\mathcal A_{\rm CFT}(u)\) with \(\sigma^v_t(x)\notin\mathcal A_{\rm CFT}(u)\), and since \(\sigma^u_t(x)\in\mathcal A_{\rm CFT}(u)\) we get
\(\sigma^v_t(x)\ne\sigma^u_t(x)\).

The displayed witnesses follow from the same fixed-point argument. If \(D_u(x)\) belonged to \(\hat{\mathcal M}(v)\), it would lie in the fixed algebra of the common dual action, forcing \(\sigma_q^u(x)=\sigma_q^v(x)\) for all \(q\), contrary to the choice of \(x\). If \(PD_u(x)P\) belonged to \(P\hat{\mathcal M}(v)P\), conjugating by \(Z(r)\) would give \(P_rD_u(x)P_r\in\hat{\mathcal M}(v)\) for every \(r\), and letting \(r\to-\infty\) would again give \(D_u(x)\in\hat{\mathcal M}(v)\).
\end{proof}

\begin{remark}[Range of the obstruction]
\label{rem:no-go-scope}
The obstruction applies to the shared-clock representation \(D_u(x)=e^{-iqK_u}(x\otimes\id)e^{iqK_u}\) with \(V(s)=e^{ish}\) and the common dual action \(\widehat\sigma_r(V(s))=e^{-irs}V(s)\). Cocycle-twisted embeddings or region-dependent clock identifications define different constructions. The representation-independent input is the absence of a vacuum-preserving conditional expectation when \(v\setminus\overline u\) contains a nonempty open arc. Pairs differing only by finitely many endpoints are instead identified by strong additivity in \Cref{ass:regional-net}(iii).
\end{remark}

\subsection{Split-State Dressing on a Fixed Boolean Family}
\label{sec:split-repair}

A split product state provides a compatible modular flow on a fixed family of separated regions. Its generator is a tensor sum in the split tensor-product representation, but not on the original CFT Hilbert space.

\begin{definition}[Split realization, product state and its generator]
\label{def:split-realization}
Let \(S=\{I_1,\dots,I_n\}\) be open arcs with pairwise disjoint closures, put \(v_S:=\bigcup_jI_j\) and \(\mathcal N_S:=\bigvee_{j}\mathcal A_{\rm CFT}(I_j)\); as is standard for a disconnected region we set \(\mathcal A_{\rm CFT}(v_S):=\mathcal N_S\), the additivity convention already used in \eqref{eq:multi-interval-split}. Under \Cref{ass:regional-net}(iii) the split property provides the unique normal isomorphism
\begin{align}
  \Theta_S:\ \mathcal N_S
  \;\xrightarrow{\ \sim\ }\;
  \overline{\bigotimes}_{j=1}^{n}\mathcal A_{\rm CFT}(I_j),
  \qquad
  \Theta_S(x_1x_2\cdots x_n)=x_1\otimes\cdots\otimes x_n
  \label{eq:split-isomorphism}
\end{align}
for \(x_j\in\mathcal A_{\rm CFT}(I_j)\), determined on such products and extended by normality. Write \(\omega_{I_j}:=\omega_{\rm BD}|_{\mathcal A_{\rm CFT}(I_j)}\) and define the faithful normal \emph{split product state}
\begin{align}
  \phi_S:=\Bigl(\bigotimes_{j=1}^{n}\omega_{I_j}\Bigr)\circ\Theta_S
  \qquad\text{on }\mathcal N_S .
  \label{eq:split-product-state}
\end{align}
Its modular automorphism group is
\begin{align}
  \sigma_t^{\phi_S}
  =\Theta_S^{-1}\circ
   \Bigl(\bigotimes_{j=1}^{n}\sigma_t^{\omega_{I_j}}\Bigr)
   \circ\Theta_S ,
  \label{eq:split-modular-flow}
\end{align}
and we let \(K_S^{\rm split}\) denote the generator of the standard-form implementation of \(\sigma^{\phi_S}\) on \(\mathcal H_{\rm CFT}\), that is \(K_S^{\rm split}:=-\log\Delta_{\phi_S}\) with \(\Delta_{\phi_S}\) the modular operator of \(\phi_S\) computed from the unique vector representative of \(\phi_S\) in the natural positive cone of
\((\mathcal N_S,\Omega_{\rm BD})\).
\end{definition}

\begin{remark}[The split modular generator]
\label{rem:split-generator-status}
\begin{enumerate}[label=(\roman*)]
\item \(K_S^{\rm split}\) becomes a tensor sum \emph{only} in the split tensor-product representation: transporting through \eqref{eq:split-isomorphism} and the associated spatial isomorphism of Hilbert spaces, \(\sigma^{\phi_S}\) becomes \(\bigotimes_j\sigma^{\omega_{I_j}}\), whose generator is \(\sum_jK_{I_j}\) acting on \(\bigotimes_j\mathcal H_{I_j}\). On \(\mathcal H_{\rm CFT}\) itself no such decomposition is available, and the expression \(\sum_jK_{I_j}\) has no meaning as a sum of operators there.
\item \(K_S^{\rm split}\ne K_{v_S}\) in general: \(\phi_S\) is a product state whereas \(\omega_{\rm BD}|_{\mathcal N_S}\) is correlated across the components. Thus \(\sigma^{\phi_S}\) differs from the vacuum Tomita modular flow of \(\mathcal A_{\rm CFT}(v_S)\).
\item The algebraic isomorphism \(\Theta_S\) in \eqref{eq:split-isomorphism} is fixed uniquely by its values on products, because those products generate \(\mathcal N_S\) as a von Neumann algebra. Spatial unitaries implementing this isomorphism, and intermediate type-\(\mathrm I\) factors used in proofs of the split property, need not be unique. That implementation freedom does not change the product state \eqref{eq:split-product-state} or its modular automorphism group.
\end{enumerate}
\end{remark}

\begin{proposition}[Isotony on one fixed finite Boolean family]
\label{prop:split-repair}
Fix \(S=\{I_1,\dots,I_n\}\) as in \Cref{def:split-realization}. For \(\varnothing\neq T\subseteq S\) put \(v_T:=\bigcup_{j\in T}I_j\), \(\mathcal N_T:=\mathcal A_{\rm CFT}(v_T)\), and \(\phi_T:=\phi_S|_{\mathcal N_T}\). Then:
\begin{enumerate}[label=(\alph*)]
\item \(\Theta_S(\mathcal N_T) =\bigl(\overline{\bigotimes}_{j\in T}\mathcal A_{\rm CFT}(I_j)\bigr) \otimes\id\), and \(\phi_T\) is the corresponding sub-product state, so \(\phi_T=(\bigotimes_{j\in T}\omega_{I_j})\circ\Theta_T\) with
\(\Theta_T:=\Theta_S|_{\mathcal N_T}\);
\item for every nonempty \(T\subseteq T'\subseteq S\), \(\sigma_t^{\phi_{T'}}\bigl|_{\mathcal N_T}=\sigma_t^{\phi_T}\) for all
\(t\in\mathbb R\);
\item consequently, writing \(\mathcal M^{\rm split}(D(v_T))\) for the shared-clock crossed product of \(\mathcal N_T\) by \(\sigma^{\phi_T}\) compressed by \(P\),
\begin{align}
  \mathcal M^{\rm split}(D(v_T))
  \subseteq\mathcal M^{\rm split}(D(v_{T'}))
  \qquad\text{whenever }\varnothing\ne T\subseteq T'\subseteq S .
  \label{eq:split-isotony}
\end{align}
\end{enumerate}
The family \(\{v_T:\varnothing\ne T\subseteq S\}\) is the finite Boolean family generated by the fixed atomic arcs \(I_1,\dots,I_n\), and \eqref{eq:split-isotony} is asserted for that family only.
\end{proposition}
\begin{proof}
(a) By \eqref{eq:split-isomorphism}, \(\Theta_S\) maps a product
\(\prod_{j\in T}x_j\) to \(\bigotimes_{j\in T}x_j\otimes\id_{S\setminus T}\);
these products generate \(\mathcal N_T\) and their images generate \(\bigl(\overline{\bigotimes}_{j\in T}\mathcal A_{\rm CFT}(I_j)\bigr) \otimes\id\), and \(\Theta_S\) is a normal isomorphism, so the two von Neumann algebras correspond. Evaluating \eqref{eq:split-product-state} on such elements and using \(\omega_{I_j}(\id)=1\) for \(j\notin T\) identifies \(\phi_T\) with the stated sub-product state.

(b) The modular group of a product state on a tensor product is the tensor product of the factor modular groups. For \(T\subseteq T'\), it therefore preserves the tensor factors indexed by \(T\) and restricts to their sub-product modular group. Transporting this statement back through \(\Theta_S\) with (a) gives (b).

(c) By (b), the pair \((\mathcal N_T\subseteq\mathcal N_{T'}, \sigma^{\phi_T},\sigma^{\phi_{T'}})\) satisfies condition (iii) of \Cref{thm:net-compatibility-criterion} in the form of \Cref{rem:criterion-general-states}, with faithful normal modular data on both sides. The implication (iii)\(\Rightarrow\)(i)\(\Rightarrow\)(ii) of that theorem gives \eqref{eq:split-isotony}.
\end{proof}

\begin{remark}[Fixed atomic family]
\label{rem:split-repair-restrictions}
The atomic arcs are fixed, and the regions form the finite Boolean family they generate. Extending \Cref{prop:split-repair} to moving endpoints or different atomic partitions would require coherent split identifications across partitions. By contrast, the vacuum-modular construction fails isotony on any family containing a pair as in \Cref{prop:net-compatibility-fails}.
\end{remark}

\section[The Ambient-Commutant Model and Its Multi-Region Behaviour]%
{The Ambient-Commutant Model and Its Multi-Region Behaviour}
\label{sec:activation-model}

\Cref{lem:ambient-representation} places every protected-gap algebra and the reference factor \(\mathcal M_{\rm ref}\) in a common ambient algebra. This allows their commutants and intersections to be studied directly.

\subsection{Defining the Model Assignment}
\label{sec:subregion-algebra-def}

The fundamental-complement principle motivates imposing commutation with the algebra assigned to the protected gap. For an admissible \(a\) with fundamental complement \(\widetilde A\) we define
\begin{align}
  \mathcal A_{\rm model}(a)
  :=\mathcal M(\widetilde A)'\cap\mathcal M_{\rm ref} ,
  \label{eq:definition-of-A}
\end{align}
where the commutant is taken in \(B(\mathcal H_{\rm CFT}\otimes L^2(\mathbb R))\), the algebra \(\mathcal M(\widetilde A)\) is the one supplied by \eqref{eq:diamond-algebra} together with \eqref{eq:spatial-trace-identity}, and \(\mathcal M(\varnothing)=\mathbb C\id\). Thus \(\mathcal A_{\rm model}(a)\) is the largest von Neumann subalgebra of the fixed reference algebra whose elements commute with the concretely represented protected-gap algebra. Unlike the generalized-entropy prescriptions of \cite{BP,BK}, no minimization enters. The suppressed dependence on the fixed reference arc \(R\) is as described in \eqref{eq:reference-covariance}.

Equation~\eqref{eq:definition-of-A} is an ambient intersection, since \(\mathcal M(\widetilde A)\) is generally not contained in \(\mathcal M_{\rm ref}\). With \(\mathcal M(\varnothing)=\mathbb C\id\), it assigns \(\mathcal M_{\rm ref}\) to an empty fundamental complement.

\begin{lemma}[Well-definedness]
\label{prop:model-status}
\(\mathcal M(\widetilde A)\) and \(\mathcal M_{\rm ref}\) are von Neumann algebras in the same \(B(\mathcal H_{\rm CFT}\otimes L^2(\mathbb R))\), so \(\mathcal A_{\rm model}(a)\) is a von Neumann subalgebra of \(\mathcal M_{\rm ref}\); and if \(\widetilde A=\varnothing\) then \(\mathcal A_{\rm model}(a)=\mathcal M_{\rm ref}\).
\end{lemma}
\begin{proof}
The first statement is \Cref{lem:ambient-representation} with the convention \(\mathcal M(\varnothing)=\mathbb C\id\), since intersections and commutants of von Neumann algebras are von Neumann algebras. The second is \((\mathbb C\id)'\cap\mathcal M_{\rm ref}=\mathcal M_{\rm ref}\).
\end{proof}

For \(a\subseteq v\), the gap theorem gives \(\widetilde v\subseteq\widetilde a\), since enlarging a region can only shrink its gaps, and isotony would follow from \(\mathcal M(\widetilde V)\subseteq\mathcal M(\widetilde A)\) by reversing commutants. That inclusion holds trivially when \(\widetilde V=\varnothing\). When both complements are nonempty, isotony remains open because the vacuum-modular dressings fail the compatibility criterion of \Cref{thm:net-compatibility-criterion}. Factoriality of \(\mathcal A_{\rm model}(a)\) is likewise unresolved. A stronger commutant recovery statement would be
\begin{align}
 \mathcal A_{\rm model}(a)'\cap\mathcal M_{\rm ref}
 =\mathcal M(\widetilde A) :
 \label{eq:commutant-recovery}
\end{align}
it requires an embedding \(\mathcal M(\widetilde A)\subseteq\mathcal M_{\rm ref}\) and a dressed commutant theorem. The results below instead concern the concrete relative embeddings fixed in \Cref{lem:ambient-representation}.

\subsection{Properness from Direct Non-Commutation}
\label{sec:direct-noncommutation}

We now prove properness when the unique protected gap overlaps the fixed reference arc \(R\). This covers every protected gap longer than \(\pi\) and every threshold gap except the antipodal placement \(R^\perp\). The argument compares the two overlapping regional algebras directly and is independent of dressed Haag duality and the split property.

For an admissible \(u\subsetneq S^1\) and \(O\in\mathcal A_{\rm CFT}(u)\) write, as in \eqref{eq:regular-core-representation},
\begin{align}
  D_u(O):=e^{-iqK_u}\,(O\otimes\id)\,e^{iqK_u} ,
  \label{eq:dressing-map}
\end{align}
so that \(\mathcal M(D(u))\) is generated by \(\{D_u(O)\}\) together with the clock translations \(V(s)=e^{ish}\), compressed by \(P\).

\begin{lemma}[Exact reduction to non-commutation]
\label{lem:exact-reduction}
For any admissible \(a\) with \(\widetilde A\ne\varnothing\),
\begin{align}
  \mathcal A_{\rm model}(a)\subsetneq\mathcal M_{\rm ref}
  \quad\Longleftrightarrow\quad
  \exists\, m\in\mathcal M_{\rm ref},\
  n\in\mathcal M(\widetilde A):\ [m,n]\ne0 .
  \label{eq:exact-reduction}
\end{align}
\end{lemma}
\begin{proof}
By \eqref{eq:definition-of-A}, \(\mathcal A_{\rm model}(a)=\mathcal M(\widetilde A)'\cap\mathcal M_{\rm ref} \subseteq\mathcal M_{\rm ref}\) always, both sides computed in the common ambient algebra. Equality holds if and only if \(\mathcal M_{\rm ref}\subseteq\mathcal M(\widetilde A)'\), that is, if and only if every element of \(\mathcal M_{\rm ref}\) commutes with every element of \(\mathcal M(\widetilde A)\). Negating gives \eqref{eq:exact-reduction}.
\end{proof}

\begin{proposition}[Non-commutation before compression]
\label{prop:pre-compression-noncommutation}
Let \(g\subsetneq S^1\) be an open arc such that \(R\cap g\) contains a nonempty open arc \(J\) with \(\overline J\subset R\cap g\). Then the unprojected crossed products \(\hat{\mathcal M}(R)\) and \(\hat{\mathcal M}(g)\) do not commute; explicitly there are \(x,y\in\mathcal A_{\rm CFT}(J)\) with
\([D_R(x),D_g(y)]\ne0\).
\end{proposition}
\begin{proof}
Since \(J\) is a nonempty proper open arc, \(\mathcal A_{\rm CFT}(J)\) is a type-\(\mathrm{III}_1\) factor and in particular non-abelian, so there are bounded \(x,y\in\mathcal A_{\rm CFT}(J)\) with \([x,y]\ne0\). Representing vectors as \(\xi(q)\in\mathcal H_{\rm CFT}\) and using that \(q\) acts on the clock factor while \(K_R,K_g\) act on the CFT factor, \(D_R(x)\) and \(D_g(y)\) are fibrewise multiplication operators,
\begin{align*}
  \bigl(D_R(x)\xi\bigr)(q)=\sigma^R_q(x)\,\xi(q),\qquad
  \bigl(D_g(y)\xi\bigr)(q)=\sigma^g_q(y)\,\xi(q),
\end{align*}
so their commutator is fibrewise multiplication by
\(c(q):=[\sigma_q^R(x),\sigma_q^g(y)]\).  At \(q=0\), \(c(0)=[x,y]\ne0\).
Strong continuity of both modular groups makes every matrix element of \(c(q)\) continuous in \(q\), so \(c(q)\ne0\) on a set of positive measure and the multiplication operator \([D_R(x),D_g(y)]\) is nonzero. Both flows preserve their respective algebras, so \(D_R(x)\) and \(D_g(y)\) are genuine elements of the two crossed products.
\end{proof}

\begin{lemma}[Compression preserves non-commutativity]
\label{lem:hardy-compression}
Let \(m,n\) be bounded fibrewise-multiplication operators on \(L^2(\mathbb R;\mathcal H_{\rm CFT})\) and let \(P=\mathbf1_{[0,\infty)}(h)\) with \(h=-i\partial_q\). If \(c:=[m,n]\ne0\), then
\begin{align}
  P\bigl(mPn-nPm\bigr)P\ne0 ,
  \label{eq:hardy-compression-nonzero}
\end{align}
as an operator on \(P\bigl(L^2(\mathbb R;\mathcal H_{\rm CFT})\bigr)\); in particular \([PmP,PnP]\ne0\).
\end{lemma}
\begin{proof}
Let \(W_N:=e^{iNq}\) for \(N>0\). From \([q,h]=i\) we get \(W_N^*hW_N=h+N\), hence \(f(h)W_N=W_Nf(h+N)\) for every Borel \(f\); in particular
\begin{align}
  (\id-P)W_N=W_N\,\mathbf 1_{(-\infty,-N)}(h).
  \label{eq:shift-identity}
\end{align}
Since \(m,n,c\) are multiplication operators in \(q\) with values in \(B(\mathcal H_{\rm CFT})\), while \(W_N\) is multiplication by the scalar function \(e^{iNq}\), they commute with \(W_N\).

The bounded-frequency subspace \(\bigcup_{L>0}\mathbf 1_{[-L,L]}(h)L^2(\mathbb R;\mathcal H_{\rm CFT})\) is dense, because spectral projections of \(h\) onto growing bounded intervals converge strongly to the identity. As \(c\ne0\), its kernel is a proper closed subspace and cannot contain a dense subspace, so we may choose \(L>0\) and \(\zeta\in\mathbf 1_{[-L,L]}(h)L^2(\mathbb R;\mathcal H_{\rm CFT})\) with \(c\zeta\ne0\); set \(\delta:=\|c\zeta\|>0\).

For \(N\ge L\) put \(\xi_N:=W_N\zeta\). Its spectral support under \(h\) lies in \([N-L,N+L]\subset[0,\infty)\), so \(P\xi_N=\xi_N\). Writing \(C:=P(mPn-nPm)P\) and using \(mPn=mn-m(\id-P)n\) together with the analogous identity for \(nPm\),
\begin{align}
  C\xi_N-c\xi_N
  = -(\id-P)c\,\xi_N-Pm(\id-P)n\,\xi_N+Pn(\id-P)m\,\xi_N .
  \label{eq:compression-identity}
\end{align}
For \(A\in\{m,n,c\}\), commutativity with \(W_N\) and \eqref{eq:shift-identity} give
\((\id-P)A\xi_N=W_N\mathbf 1_{(-\infty,-N)}(h)\,A\zeta\), so
\(\|(\id-P)A\xi_N\|=\|\mathbf 1_{(-\infty,-N)}(h)\,A\zeta\|\to0\) as
\(N\to\infty\), since spectral projections onto a half-line receding to \(-\infty\) converge strongly to zero on the fixed vector \(A\zeta\). Hence, from \eqref{eq:compression-identity} and \(\|P\|\le1\),
\begin{align*}
  \|C\xi_N-c\xi_N\|\le\|(\id-P)c\xi_N\|
  +\|m\|\,\|(\id-P)n\xi_N\|+\|n\|\,\|(\id-P)m\xi_N\|
  \ \longrightarrow\ 0 .
\end{align*}
Meanwhile \(c\xi_N=cW_N\zeta=W_Nc\zeta\), so \(\|c\xi_N\|=\delta\) for every \(N\) by unitarity of \(W_N\). Choosing \(N\) so large that \(\|C\xi_N-c\xi_N\|<\delta/2\) gives \(\|C\xi_N\|\ge\delta/2>0\), whence \(C\ne0\). Finally \([PmP,PnP]=P(mPn-nPm)P=C\).
\end{proof}

\begin{corollary}[Properness for an overlapping protected gap]
\label{cor:properness-a1}
Let \(a\) be admissible with nonempty fundamental complement, and let \(g\) be its unique protected gap. If \(R\cap g\) contains a nonempty open arc \(J\) with \(\overline J\subset R\cap g\), then
\begin{align}
  \mathcal A_{\rm model}(a)\subsetneq\mathcal M_{\rm ref} .
  \label{eq:cor-properness}
\end{align}
The overlap hypothesis is automatic if \(|g|>\pi\). If \(|g|=\pi\), it holds unless \(g=R^\perp\). In particular, for \(a_1=(0,\delta)\) with \(0<\delta<\pi\), the protected gap \(g=a_1^\perp=(\delta,2\pi)\) has length \(2\pi-\delta>\pi\), and therefore \(\mathcal A_{\rm model}(a_1)\subsetneq\mathcal M_{\rm ref}\). The result concerns the relative embeddings fixed in \Cref{lem:ambient-representation}; factoriality of the proper subalgebra is left open.
\end{corollary}
\begin{proof}
By hypothesis choose \(J\) with \(\overline J\subset R\cap g\). \Cref{prop:pre-compression-noncommutation} supplies bounded \(m=D_R(x)\) and \(n=D_g(y)\) with \([m,n]\ne0\), and both are bounded fibrewise-multiplication operators, so \Cref{lem:hardy-compression} gives
\[
 [PmP,PnP]=P(mPn-nPm)P\ne0 .
\]
Here \(PmP\in\mathcal M_{\rm ref}\) and \(PnP\in\mathcal M(\widetilde A_1)\), since \(P\) belongs to both represented cores by \Cref{lem:ambient-representation}. \Cref{lem:exact-reduction} now gives \eqref{eq:cor-properness}.

It remains to justify the stated automatic cases. Since \(|R|=\pi\), an open arc \(g\) with \(|g|>\pi\) cannot be disjoint from \(R\); their intersection is then a nonempty open set and contains an arc with closure inside it. If \(|g|=\pi\), two open half-circles are disjoint only when they are regular-open complements, so the same conclusion holds unless \(g=R^\perp\). The displayed family \(a_1=(0,\delta)\) has \(|a_1^\perp|=2\pi-\delta>\pi\), proving the last statement.
\end{proof}

\subsection{Gap Protection in the Model}
\label{sec:algebraic-gap-protection}

\needspace{15\baselineskip}
\begin{proposition}[Gap protection in the model]
\label{prop:algebraic-gap-protection}
Let \(a\) be admissible with gaps \(b_k\) of angular lengths \(\beta_k\). Then
\begin{enumerate}[label=(\roman*)]
\item the protected algebra is
    \[
      \mathcal M(\widetilde A)=
      \begin{cases}
        \mathbb C\id,&\max_k\beta_k<\pi,\\[2pt]
        \mathcal M(D(b_*)),&\beta_*\geq\pi,
      \end{cases}
    \]
where \(b_*\) is the unique protected gap in the second case; and
\item if \(\widetilde A=\varnothing\), then \(\mathcal A_{\rm model}(a)=\mathcal M_{\rm ref}\).
\end{enumerate}
\end{proposition}

\begin{proof}
(i) combines \Cref{thm:gap-protection}, the convention \(\mathcal M(\varnothing)=\mathbb C\id\), and \eqref{eq:diamond-algebra}. (ii) is the elementary identity \((\mathbb C\id)'\cap\mathcal M_{\rm ref}=\mathcal M_{\rm ref}\), that is, the empty-complement branch of the definition \eqref{eq:definition-of-A}.
\end{proof}

\subsection{Activation: A Discontinuity in the Assigned Algebra}
\label{sec:activation}

\subsubsection{The Hologram of an Unprotected Region}
\label{sec:refined-wedge-empty}

In the framework of \cite{BP,BK}, accessibility conditions based on the smooth conditional max-entropy assign max- and min-holograms to a gravitating wedge \(A\); when they coincide we write \(E(A)\) for the common hologram.

\begin{definition}[Hologram trace]
\label{def:hologram-trace}
For a spacetime wedge \(A\), we define its \emph{hologram trace} on the Cauchy slice \(\Sigma_0\) by
\[
e_{\Sigma_0}(A):=E(A)\cap\Sigma_0.
\]
For a spatial region \(a\subset\Sigma_0\), we use the shorthand
\[
e(a):=e_{\Sigma_0}\bigl(D(a)\bigr)
     =E\bigl(D(a)\bigr)\cap\Sigma_0.
\]
Thus, if \(A=D(a)\), then \(e(a)=e_{\Sigma_0}(A)\).
\end{definition}

The following endpoint result will be used.

\begin{proposition}[Bousso--Kaya]
\label{prop:refined-wedge-empty}
Let \(a\subseteq S^1\) be admissible with \(\widetilde A=\varnothing\), and let \(A=D(a)\). Then the max- and min-holograms of \(A\) coincide and exhaust the spacetime,
\begin{align*}
  E_{\max}(A)=E_{\min}(A)=M ,
\end{align*}
so its hologram trace is the whole Cauchy slice, \(e(a)=S^1\).
\end{proposition}

\begin{proof}
Theorem~23 of \cite{BK} gives \(E_{\max}(A)=E_{\min}(A)=M\) when the fundamental complement vanishes. Their wedge \(a\) and complement \(\tilde a\) are denoted here by \(A\) and \(\widetilde A\). Intersecting the common hologram with \(\Sigma_0\) gives the statement about its trace.
\end{proof}

\begin{remark}[Semiclassical cross-check]
\label{rem:s1-entropy-crosscheck}
In the semiclassical Bunch--Davies regime, the generalized entropy
\[
S_{\rm gen}(r)=\frac{\Phi(\partial r)}{4G_2}+S_{\rm CFT}(r)
\]
also favors \(r=S^1\), the hologram trace of \Cref{prop:refined-wedge-empty}. The whole circle has \(\Phi(\partial S^1)=0\), and since \(S^1\) has empty complement \(\ket{\rm BD}\) restricts to a vector state on \(B(\mathcal H_{\rm CFT})\), hence a pure one, so \(S_{\rm gen}(S^1)=0\) exactly. For \(r\subsetneq S^1\) with gaps parametrically larger than the cutoff, the single-interval vacuum entanglement-entropy formulas \cite{HolzheyWilczek,CalabreseCardy} give \(S_{\rm CFT}(r)>0\), and if \(\Phi/4G_2\ge0\), then \(S_{\rm gen}(r)>0\).

This check assumes the vacuum, a nonnegative dilaton contribution, and gaps larger than the cutoff. Multi-interval entropy depends on the full operator content of the CFT, so the argument is not an alternative derivation of \Cref{prop:refined-wedge-empty}.
\end{remark}

\subsubsection{Two-Arc Activation}
\label{sec:two-interval-activation}

Start with a single short arc (see \Cref{fig:activation}),
\begin{align*}
  a_1 = (0,\delta),\qquad 0<\delta\ll 1 .
\end{align*}
Its regular-open complement is a single protected gap of angular length
\(2\pi-\delta>\pi\), so
\begin{align*}
  \widetilde A_1=D\bigl(a_1^\perp\bigr),\qquad
  \mathcal A_{\rm model}(a_1)\subsetneq\mathcal M_{\rm ref},
\end{align*}
the second statement being \Cref{cor:properness-a1}. Now add a second, short arc on the opposite side of the circle,
\begin{align*}
  a_2 = (\pi,\pi+\varepsilon),\qquad 0<\varepsilon\ll 1 .
\end{align*}
The union \(a=a_1\cup a_2\) is admissible for \(0<\delta<\pi\) and \(0<\varepsilon<\pi\), and leaves two gaps,
\begin{align*}
  b_1=(\delta,\pi),\quad |b_1|=\pi-\delta<\pi,\qquad
  b_2=(\pi+\varepsilon,2\pi),\quad |b_2|=\pi-\varepsilon<\pi .
\end{align*}
Both are below the threshold, so \(\widetilde A=\varnothing\) and, by \Cref{prop:algebraic-gap-protection}(ii), \(\mathcal A_{\rm model}(a)=\mathcal M_{\rm ref}\): the addition of the short antipodal arc removes the model's commutation restriction.

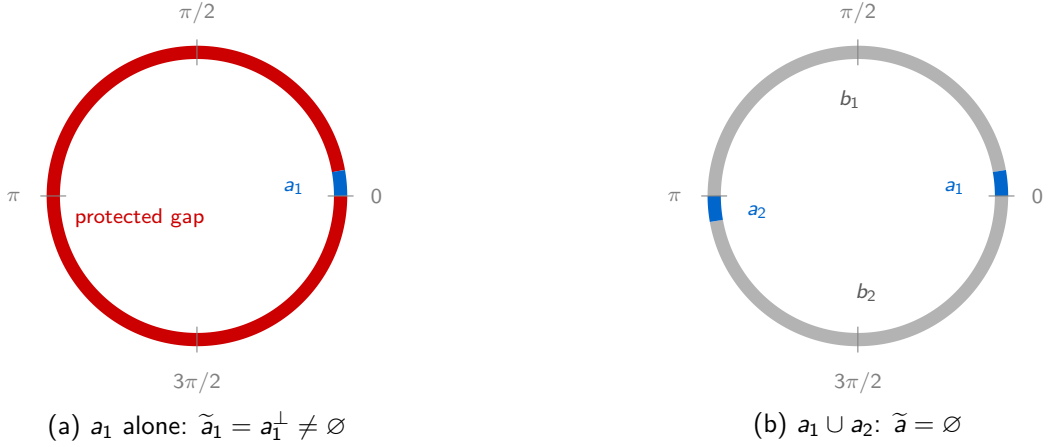
\begin{figure}[t]
\centering
{\sffamily\sansmath
\begin{tikzpicture}[scale=0.95]
  \begin{scope}[xshift=-4.6cm]
    \draw[figred, line width=5pt] plot[domain=10:360,samples=200] (\x:2);
    \draw[figblue, line width=5pt] plot[domain=0:10,samples=10] (\x:2);
    \foreach \ang in {0,90,180,270}{\draw[gray] (\ang:2.18) -- (\ang:1.82);}
    \node[gray] at (0:2.5) {\scriptsize $0$};
    \node[gray] at (90:2.55) {\scriptsize $\pi/2$};
    \node[gray] at (180:2.55) {\scriptsize $\pi$};
    \node[gray] at (270:2.6) {\scriptsize $3\pi/2$};
    \node[figblue] at (5:1.35) {\scriptsize $a_1$};
    \node[figred] at (195:1.15) {\scriptsize $\qquad$protected gap};
    \node at (0,-3.2) {\small (a) $a_1$ alone:
      $\widetilde a_1=a_1^\perp\neq\varnothing$};
  \end{scope}
  \begin{scope}[xshift=4.6cm]
    \draw[gray!60, line width=5pt] plot[domain=10:180,samples=100] (\x:2);
    \draw[gray!60, line width=5pt] plot[domain=190:360,samples=100] (\x:2);
    \draw[figblue, line width=5pt] plot[domain=0:10,samples=10] (\x:2);
    \draw[figblue, line width=5pt] plot[domain=180:190,samples=10] (\x:2);
    \foreach \ang in {0,90,180,270}{\draw[gray] (\ang:2.18) -- (\ang:1.82);}
    \node[gray] at (0:2.5) {\scriptsize $0$};
    \node[gray] at (90:2.55) {\scriptsize $\pi/2$};
    \node[gray] at (180:2.55) {\scriptsize $\pi$};
    \node[gray] at (270:2.6) {\scriptsize $3\pi/2$};
    \node[figblue] at (5:1.35) {\scriptsize $a_1$};
    \node[figblue] at (189:1.42) {\scriptsize $a_2$};
    \node[gray!70!black] at (95:1.35) {\scriptsize $b_1$};
    \node[gray!70!black] at (275:1.35) {\scriptsize $b_2$};
    \node at (0,-3.2) {\small (b) $a_1\cup a_2$: $\widetilde a=\varnothing$};
  \end{scope}
\end{tikzpicture}
\caption{The activation effect. \textbf{(a)}~A short arc \(a_1=(0,\delta)\) leaves one protected gap \(a_1^\perp=(\delta,2\pi)\), and \(\mathcal A_{\rm model}(a_1)\) is a proper von Neumann subalgebra of \(\mathcal M_{\rm ref}\) by \Cref{cor:properness-a1}. \textbf{(b)}~A short antipodal arc \(a_2=(\pi,\pi+\varepsilon)\) splits that gap into the two open arcs \(b_1=(\delta,\pi)\) and \(b_2=(\pi+\varepsilon,2\pi)\), both of angular length below the threshold \(\pi\). The fundamental complement is then empty and the definition \eqref{eq:definition-of-A} assigns the whole reference algebra. Arc widths are exaggerated for visibility.}
}
\label{fig:activation}
\end{figure}

\begin{definition}[Activation]
\label{def:activation}
For the explicit family in \Cref{thm:activation}, \emph{activation} is the discontinuous change from the proper subalgebra assigned to \(a_1=(0,\delta)\) to the whole reference algebra assigned after an arbitrarily short antipodal arc is added. The first statement follows from \Cref{cor:properness-a1}, while the empty-complement endpoint is fixed by definition in \Cref{prop:algebraic-gap-protection}(ii).
\end{definition}

\begin{theorem}[Activation]
\label{thm:activation}
There exist admissible regions \(a_1\subsetneq a\subsetneq S^1\), with \(a\setminus a_1\) of arbitrarily small angular length, such that
\begin{align}
  \mathcal A_{\rm model}(a_1)\subsetneq\mathcal M_{\rm ref}
  \qquad\text{and}\qquad
  \mathcal A_{\rm model}(a)=\mathcal M_{\rm ref} .
  \label{eq:activation-unconditional}
\end{align}
\end{theorem}

\begin{proof}
Take \(a_1=(0,\delta)\) with \(0<\delta<\pi\) and \(a=a_1\cup a_2\) with
\(a_2=(\pi,\pi+\varepsilon)\), \(0<\varepsilon<\pi\), so
\(a\setminus a_1=a_2\) has angular length \(\varepsilon\), which may be taken arbitrarily small. \Cref{cor:properness-a1} gives the first statement. The two gaps of \(a\), namely \(b_1=(\delta,\pi)\) and \(b_2=(\pi+\varepsilon,2\pi)\), have angular lengths \(\pi-\delta\) and \(\pi-\varepsilon\), both strictly less than \(\pi\) for \emph{any} \(\delta,\varepsilon>0\). Hence \(\widetilde A=\varnothing\) by \Cref{thm:gap-protection}, regardless of how small \(\varepsilon\) is, and \Cref{prop:algebraic-gap-protection}(ii) gives the second equality.
\end{proof}

\begin{remark}[Continuum limit]
\label{rem:uv-cutoff}
Haag--Kastler nets assign algebras to arbitrarily small regions, as in the crossed-product literature \cite{CLPW,Witten2022}. With a fixed UV cutoff, the same result holds for any \(\varepsilon\) between that cutoff and \(\pi\), since both new gaps remain strictly below the threshold.
\end{remark}

\section{Discussion}
\label{sec:discussion}

The gap theorem turns the fundamental-complement prescription of \cite{BP,BK} into a simple geometric rule on global \(\mathrm{dS}_2\). A complementary gap is protected exactly at and above angular size \(\pi\), the threshold at which it supports a timelike curve of infinite proper time in both directions. Since only the largest gap matters, a short antipodal component can split one protected gap into two unprotected ones. For the two-arc family of \Cref{thm:activation}, this geometric transition changes the assigned algebra from a proper subalgebra of \(\mathcal M_{\rm ref}\) to the full reference algebra.

The modular analysis identifies a separate constraint on regional constructions. The continuous cores built from the regions' vacuum modular flows are nested exactly when the larger modular flow restricts to the smaller one (\Cref{thm:net-compatibility-criterion}). Vacuum correlations violate this condition whenever the larger region adds a nonempty spacelike arc (\Cref{prop:net-compatibility-fails}). Regional dressings must therefore be coordinated under inclusion, rather than chosen independently. The split product state of \Cref{prop:split-repair} demonstrates such coordination on a fixed finite Boolean family, although comparing different atomic partitions requires additional data.

The construction extends the static-patch perspective of \cite{CLPW} to multi-arc regions without changing the algebraic types involved. The matter algebra \(\mathcal A_{\rm CFT}(R)\) is type \(\mathrm{III}_1\), and \(\mathcal M_{\rm ref}\) is its type-\(\mathrm{II}_1\) positive-clock-spectrum corner. Activation is a discontinuity of the region-to-algebra assignment built from these objects. This provides a concrete example of the kind of map considered in the algebraic generalized-entanglement-wedge proposal of \cite{SahuGEWAlgebras}. Establishing its expected isotony and subfactor structure would require a compatible dressed net and the commutant identification \eqref{eq:commutant-recovery}.

A semiclassical target for that identification already exists. If the gravitational modular Hamiltonian, with the same convention \(K_A=-\log\Delta_A\) as in \eqref{eq:modular-convention}, takes the usual fixed-phase form
\begin{align}
  K_A=\hat{\mathscr A}(\partial E_0)+K^{\rm EFT}_{E_0}+O(G_2)
  \label{eq:semiclassical-modular}
\end{align}
in a code subspace, with a state-independent edge operator \(\hat{\mathscr A}\), then the edge term cancels in relative-entropy differences and Petz sufficiency \cite{Petz86,Petz88,JLMS,DHW} supplies a reconstruction map from the boundary effective-field-theory algebra on \(E(A)\) into the physical gravitational algebra. Identifying its target with \(\mathcal A_{\rm model}(a)\) requires an isomorphism between that physical algebra and the present common-clock model:
\begin{align}
  \mathcal A_{\rm semi}\bigl(E(A)\bigr)
  \longrightarrow\mathcal A_{\rm model}(a)
  \label{eq:matching-condition}
\end{align}
and should be viewed as the desired matching condition.

The canonical analysis of \Cref{app:adm} supplies the gravity-plus-matter constraints, the endpoint balance \eqref{eq:boundary-term}, the classical relation \(Q^R_{\rm JT}=H^R_{\rm mat}\), and the York pair \(\bigl(K,a\Phi/(8\pi G_2)\bigr)\). The modular identity \(K_R=2\pi H_{\rm stat}^{\rm ch}\), by contrast, belongs to the chiral model. The JT constraint couples to the one-sided charge \(H^R_{\rm mat}\) in the full CFT, while its global generator is \(H_{\rm stat}^{\rm full}=H_{\rm mat}^{R}+H_{\rm mat}^{R^\perp}\). The common Killing field makes the normalizations comparable, but a quantum matching requires a reduced gravitational clock constraint. The York-time constructions of \cite{ParrikarSake,BlommaertChen} provide natural starting points for this step. In the asymptotically bounded setting, \cite{PeningtonWitten2023} gives the complementary canonical construction of the JT boundary algebra.

Two extensions appear especially natural. The first is a non-chiral, anomaly-free matter realization with regional embeddings coherent under inclusion. The second is the higher-dimensional problem suggested by \Cref{prop:dSd-threshold}: determining the causal completion of a general hemisphere-containing component and its operator-algebraic counterpart.

\section*{Acknowledgement}
The author would like to thank Dima Arinkin and Gary Shiu for fruitful discussions related to this work.

\textbf{Funding.} The author declares that no funding, grants, or other support were received during the preparation of this manuscript.

\textbf{Data Availability.} No datasets were generated or analysed during the current study; the schematic figures were produced directly from the closed-form expressions given in the text.

\appendix

\section{Conventions and Normalizations}
\label{app:normalization-ledger}

The conventions used in the algebraic construction are collected here.
\begin{align*}
 \Delta_u&:\ \text{Tomita modular operator of }
   (\mathcal A_{\rm CFT}(u),\Omega_{\rm BD}),\\
 K_u&=-\log\Delta_u,\\
 \sigma_t^u(x)&=\Delta_u^{it}x\Delta_u^{-it}
   =e^{-itK_u}x\,e^{itK_u},\\
 H_{\rm stat}^{\rm ch}&:\ \text{generator of the M\"obius subgroup preserving
   }R\text{ in the abstract chiral model},\\
 K_R&=2\pi H_{\rm stat}^{\rm ch}
   \quad\text{(\Cref{lem:static-modular-relation}, \Cref{rem:sign-check})},\\
 H_{\rm stat}^{\rm full}
   &=\int_{S^1}\ed\chi\,\cos\chi\,T_{\eta\eta}(0,\chi)
   \quad\text{(global charge of the separate full CFT)},\\
 H_{\rm mat}^R&=\int_R\ed\chi\,\cos\chi\,T_{\eta\eta}(0,\chi)
   \quad\text{(one-sided matter charge; requires
   \Cref{ass:stress-tensor})},\\
 H_{\rm stat}^{\rm full}&\ne H_{\rm mat}^R
   \quad\text{(\Cref{rem:hstat-neq-hmat})},\\
 Q^R_{\rm JT}&=-(\Phi_++\Phi_-)/(8\pi G_2),\qquad
 Q^R_{\rm JT}=H^R_{\rm mat}\ \text{classically on the constraint surface},\\
 q_{\rm clk}&:\ \text{dimensionless auxiliary modular-clock coordinate},\\
 p_{\rm clk}&=-i\partial_{q_{\rm clk}},\quad
 [q_{\rm clk},p_{\rm clk}]=i,
 \quad q:=q_{\rm clk},\quad h:=p_{\rm clk},\\
 D_u(x)&=\sigma_q^u(x),\qquad V(s)=e^{ish},\\
 \widehat\sigma_r(D_u(x))&=D_u(x),\qquad
 \widehat\sigma_r(V(s))=e^{-irs}V(s),\qquad Z(r)=e^{irq},\\
 P&=\mathbf1_{[0,\infty)}(h),\qquad
 P_r=Z(r)PZ(r)^*=\mathbf1_{[r,\infty)}(h),\\
 \rho_u^{is}&=V(s),\qquad \log\rho_u=h,\\
 \widehat\Tr_u(x)&=\widehat\varphi_u\bigl(
   \rho_u^{-1/2}x\rho_u^{-1/2}\bigr),\qquad
   x\in\hat{\mathcal M}(u)_+,\\
 \widehat\Tr_u\circ\widehat\sigma_r&=e^{-r}\widehat\Tr_u,\qquad
 \widehat\Tr_u(P)=1,\qquad
 \tau_{u,P}(x)=\widehat\Tr_u(x)/\widehat\Tr_u(P).
\end{align*}
The parameters \(t,q,r,s\) are dimensionless. Physical static time is \(\tau\), with \(\sigma_t^R=\alpha_{-2\pi t}^{\rm stat}\), so the fibre value \(q\) corresponds to \(\tau=-2\pi q\). The angular symbols \(\beta,\beta_k\) denote gap lengths only. The auxiliary \(h\) above is unrelated to the ADM spatial metric \(h\) of \Cref{app:adm}, to the formal endpoint charge \(Q_{\rm JT}^R\), and to the candidate York momentum \(a\Phi/(8\pi G_2)\). The value \(\widehat\Tr_u(P)=1\) uses the normalized regional vacuum state together with the Haar convention \(\ed\mu(s)=\ed s/(2\pi)\) and the Fourier pair fixed in the proof of \Cref{prop:finite-trace}. The symmetric trace formula is an extended-positive-cone identity when \(\rho_u\) is unbounded. An abstract semifinite trace may still be rescaled (\Cref{rem:trace-scale}), and the traces for different \(u\) are not asserted to be restrictions of one ambient trace (\Cref{rem:common-projection-traces}).

\section{Geometric Form of the Inclusion Obstruction}
\label{app:geometric-isotony-check}

For nested single arcs the obstruction of \Cref{prop:net-compatibility-fails} also has a direct geometric proof. Map \(v\) to \((0,\infty)\); by \Cref{lem:geometric-modular-flow} its vacuum modular group acts as \(x\mapsto e^{2\pi t}x\). If \(u=(a_0,b_0)\) is compactly contained in \(v\), then for sufficiently large \(t\) the arcs \(u\) and \(e^{2\pi t}u\) are disjoint, so \Cref{lem:statistical-independence} makes their algebras intersect in \(\mathbb C\id\); being type-\(\mathrm{III}_1\) factors, they cannot coincide. If \(u\) shares one endpoint with \(v\), a nonzero dilation instead produces two strictly nested arcs; their algebras are still distinct, since equality would place the nontrivial algebra of an arc in their set-theoretic difference both inside the smaller factor and, by locality, inside its commutant, contradicting factoriality. Hence \(\sigma_t^v(\mathcal A_{\rm CFT}(u))\ne\mathcal A_{\rm CFT}(u)\) for some \(t\), in agreement with the conditional-expectation proof.

\needspace{6\baselineskip}
\section{Canonical Dilaton-Gravity Analysis and the Auxiliary Clock}
\label{app:adm}

This appendix derives the canonical content of \eqref{eq:jt-action} and separates it from the auxiliary crossed-product construction. \Cref{ass:stress-tensor} is in force throughout. Write \(g:=8\pi G_2\) and \(a:=\sqrt h\), where \(h\) is the ADM spatial metric --- in this appendix only, \(h\) never denotes the auxiliary clock momentum. We work first on a general spacelike slice and retain \(L\).

\subsection{ADM Variables and Curvature Identity}
\label{app:adm-setup}

Write the metric in lapse-shift form,
\begin{align}
  \ed s^2 = -N^2\,\ed t^2 + a^2(\ed\chi+N^\chi\,\ed t)^2 ,
  \label{eq:adm-metric}
\end{align}
with future unit normal \(n=N^{-1}(\partial_t-N^\chi\partial_\chi)\). Fix the sign of the extrinsic curvature by
\begin{align}
  K = \frac{1}{Na}\Bigl(\dot a-(aN^\chi)'\Bigr).
  \label{eq:extrinsic-curvature}
\end{align}
For this convention, direct evaluation of the Ricci scalar gives
\begin{align}
  \sqrt{-g}\,R
  =-2\partial_\chi\!\left(\frac{N'}a\right)
   -2\partial_\chi(aKN^\chi)+2\partial_t(aK).
  \label{eq:gauss-codazzi}
\end{align}
The signs can be checked without a reference: for \(N=1,\ N^\chi=0,\ a=a(t)\), the metric \eqref{eq:adm-metric} has \(R=2\ddot a/a\) and \(\sqrt{-g}=a\), so \(\sqrt{-g}R=2\ddot a\), while \eqref{eq:gauss-codazzi} gives \(K=\dot a/a\) and \(2\partial_t(aK)=2\ddot a\). Formulae written with the opposite convention for \(K\) have the opposite time and shift divergences. Canonical formulations of two-dimensional dilaton gravity provide useful comparisons \cite{NavarroSalasTalavera,KucharRomanoVaradarajan}, but the calculation below follows directly from \eqref{eq:jt-action}.

\subsection{Bulk Lagrangian and Boundary Terms}
\label{app:adm-lagrangian}

The bulk coefficient in \eqref{eq:jt-action} is \(1/(2g)\). Substituting \eqref{eq:gauss-codazzi} and applying the product rule gives
\begin{align}
 \frac1g\Phi\partial_t(aK)
   &=\frac1g\partial_t(\Phi aK)-\frac1g aK\dot\Phi,
   \nonumber\\
 -\frac1g\Phi\partial_\chi(aKN^\chi)
   &=-\frac1g\partial_\chi(\Phi aKN^\chi)
     +\frac1g\Phi'aKN^\chi,
   \nonumber\\
 -\frac1g\Phi\partial_\chi\!\left(\frac{N'}a\right)
   &=\frac1g\partial_\chi\!\left(
       -\frac{\Phi N'}a+\frac{N\Phi'}a\right)
     -\frac1gN\partial_\chi\!\left(\frac{\Phi'}a\right).
 \label{eq:adm-ibp-lines}
\end{align}
The first total derivative is cancelled on spacelike time boundaries by the GHY term of \eqref{eq:jt-action} with the same normal convention. Before any additional spatial, null, or corner action is chosen, the spatial total derivative is
\begin{align}
 B_{\rm raw}
 =\frac1g\left[
 -\frac{\Phi N'}a+\frac{N\Phi'}a-\Phi aKN^\chi
 \right]_{\partial\Sigma}.
 \label{eq:adm-raw-boundary}
\end{align}
It vanishes after summing endpoints on a closed \(S^1\). On a cut-open static patch it must instead be combined with the boundary and joint terms appropriate to the chosen variational problem; in particular the horizons are null boundaries, and a timelike GHY term cannot simply be assigned to them.

The gravitational bulk Lagrangian density is
\begin{align}
  \mathcal L_{\rm grav}
  =\frac1g\left[
  -aK(\dot\Phi-N^\chi\Phi')
  -N\partial_\chi\!\left(\frac{\Phi'}a\right)
  -\frac{Na\Phi}{L^2}\right],
  \label{eq:adm-lagrangian}
\end{align}
equivalently
\begin{align}
  \mathcal L_{\rm grav}
  =\frac1g\left[
  -\frac1N\bigl(\dot a-(aN^\chi)'\bigr)
     (\dot\Phi-N^\chi\Phi')
  -N\partial_\chi\!\left(\frac{\Phi'}a\right)
  -\frac{Na\Phi}{L^2}\right].
  \label{eq:adm-lagrangian-expanded}
\end{align}

\subsection{Momenta and Primary Constraints}
\label{app:adm-momenta}

Write \(\mathsf A:=\dot a-(aN^\chi)'\) and \(\mathsf B:=\dot\Phi-N^\chi\Phi'\). The momenta are
\begin{align}
  P_\Phi&=-\frac{\mathsf A}{gN}=-\frac{aK}{g},
  \label{eq:p-phi}\\
  P_a&=-\frac{\mathsf B}{gN},\qquad
  P_N=0,\qquad P_{N^\chi}=0,
  \label{eq:p-sqrth}
\end{align}
the last two being primary constraints. The velocities are
\begin{align}
  \dot a=(aN^\chi)'-gNP_\Phi,\qquad
  \dot\Phi=N^\chi\Phi'-gNP_a .
  \label{eq:velocity-inversion}
\end{align}

\subsection{Gravity Plus Matter Constraints}
\label{app:adm-legendre}

Using \eqref{eq:velocity-inversion} in the Legendre transform and integrating \(P_a(aN^\chi)'\) once gives
\begin{align}
 I_{\rm grav}
 =\int \ed t\,\ed\chi\,
 \bigl(P_\Phi\dot\Phi+P_a\dot a
 -N\mathcal C_{\rm grav}-N^\chi\mathcal C_{\chi,\rm grav}\bigr)
+I_{\partial}
  \label{eq:hamiltonian-decomposition}
\end{align}
with
\begin{align}
 \mathcal C_{\rm grav}
 &=-gP_\Phi P_a+\frac1g\left[
   \partial_\chi\!\left(\frac{\Phi'}a\right)
   +\frac{a\Phi}{L^2}\right],
 \label{eq:hamiltonian-constraint-density}\\
 \mathcal C_{\chi,\rm grav}
 &=P_\Phi\Phi'-aP_a' .
 \label{eq:momentum-constraint-density}
\end{align}
Here \(I_\partial\) includes \eqref{eq:adm-raw-boundary}, the shift endpoint term, and whatever boundary or corner functional makes the chosen variational problem differentiable.

For canonical matter variables \((\varphi^{\mathsf A},\pi_{\mathsf A})\) define the ADM densities by
\begin{align}
 \mathcal L_{\rm mat}
 =\pi_{\mathsf A}\dot\varphi^{\mathsf A}-N\mathcal C_{\rm mat}
  -N^\chi\mathcal C_{\chi,\rm mat},\qquad
 \mathcal C_{\rm mat}=aT_{\mu\nu}n^\mu n^\nu .
 \label{eq:matter-adm}
\end{align}
This fixes the sign of \(\mathcal C_{\chi,\rm mat}\) through variation with respect to \(N^\chi\) and avoids a convention-dependent sign for the momentum flux. The complete canonical action also contains \(-\lambda_NP_N-\lambda_\chi P_{N^\chi}\); preservation of the two primary constraints gives the secondary constraints
\begin{align}
 \mathcal C_{\rm grav}+\mathcal C_{\rm mat}=0,\qquad
 \mathcal C_{\chi,\rm grav}+\mathcal C_{\chi,\rm mat}=0 .
 \label{eq:total-adm-constraints}
\end{align}
No property of the CFT state enters this derivation; the matter term follows from metric variation alone.

\subsection{Evaluation on the Time-Reflection-Symmetric Slice}
\label{app:adm-slice}

For the coordinates of \eqref{eq:global-ds-metric}, \(t=\eta\), the background slice has \(N_0=a_0=L\) and \(N^\chi=0\). Time reflection gives \(K=0\) and \(\dot\Phi=0\), hence \(P_\Phi=P_a=0\). Since \(n=L^{-1}\partial_\eta\),
\begin{align}
 \mathcal C_{\rm grav}\big|_{\rm sym}
 &=\frac1{gL}(\Phi''+\Phi),&
 \mathcal C_{\rm mat}\big|_{\rm sym}
 &=\frac1L T_{\eta\eta}.
  \label{eq:h-sym}
\end{align}
The total constraint is therefore exactly \eqref{eq:adm-constraint-density}. This also follows from the \(\eta\eta\) component of \eqref{eq:dilaton-eom}, which provides an independent sign check: at \(\eta=0\) all Christoffel symbols of \eqref{eq:global-ds-metric} vanish and \(e^{2\omega}=L^2\), so the left-hand side of \eqref{eq:dilaton-eom} reduces to \(-(\partial_\chi^2\Phi+\Phi)\), giving
\(\tfrac1g(\partial_\chi^2\Phi+\Phi)+T_{\eta\eta}=0\).

\subsection{Killing Smearing, Endpoints, and the Classical Charge Relation}
\label{app:adm-smearing}

The de Sitter boost that is future directed in \(D(R)\) has
\(\xi=\cos\chi\,\partial_\eta\) on \(\Sigma_0\), as in
\eqref{eq:killing-field}. Its canonical lapse is \(N_\xi=-n_\mu\xi^\mu=L\cos\chi\), not \(\cos\chi\) unless \(L=1\). Smearing \eqref{eq:total-adm-constraints} over \(R=(-\pi/2,\pi/2)\) gives
\begin{align}
 0&=\int_R \ed\chi\,N_\xi
  (\mathcal C_{\rm grav}+\mathcal C_{\rm mat})\big|_{\rm sym}
 \nonumber\\
 &=\frac1g\int_R \ed\chi\,\cos\chi\,(\Phi''+\Phi)
   +\int_R \ed\chi\,\cos\chi\,T_{\eta\eta} .
 \label{eq:smeared-total-constraint}
\end{align}
Integrating the first term by parts twice,
\begin{align}
 \int_R \ed\chi\,\cos\chi\,(\Phi''+\Phi)
 =\left[\cos\chi\,\Phi'+\sin\chi\,\Phi\right]_{-\pi/2}^{\pi/2}
 =\Phi_++\Phi_- ,
  \label{eq:app-boundary-term}
\end{align}
with \(\Phi_\pm=\Phi(0,\pm\pi/2)\). This proves \eqref{eq:boundary-term}. The cancellation is independent of \(L\): the background factor in \(N_\xi\) cancels the \(1/L\) in the constraint density.

If the endpoint values of \(\Phi\) are fixed, their variations vanish and \eqref{eq:boundary-term} is simply an integrated constraint. If they are allowed to vary, variation of the bulk smeared gravitational constraint contains
\begin{align}
 \frac1g\bigl[(N_\xi/L)\,\delta\Phi'
       -(N_\xi/L)'\,\delta\Phi\bigr]_{\partial R}
 =\frac1g\bigl(\delta\Phi_++\delta\Phi_-\bigr),
 \label{eq:endpoint-variation}
\end{align}
so a differentiable generator adds \(Q_{\rm JT}^R=-(\Phi_++\Phi_-)/g\), in which case \eqref{eq:boundary-term} reads \(Q_{\rm JT}^R=H_{\rm mat}^R\). This is \Cref{prop:classical-charge-relation}. The sign changes with the generator orientation, and a static-patch Hamiltonian additionally requires the appropriate null-boundary and corner variational principle.

The endpoints \((\eta,\chi)=(0,\pm\pi/2)\) are the two bifurcation points on the spatial slice. The chronological future and past tips of \(D(R)\) are instead \((\pm\pi/2,0)\) in the conformal completion. Sharing horizon generators does not identify these four points.

\subsection{Extrinsic Curvature as a Clock Candidate}
\label{app:adm-york}

Equation~\eqref{eq:p-phi} makes the extrinsic curvature a natural local clock candidate. Set
\begin{align}
 T_Y:=K,\qquad
 \Pi_Y:=\frac{a\Phi}{g},\qquad
 \Pi_a:=P_a+\frac{\Phi K}{g} .
 \label{eq:york-canonical-pair}
\end{align}
The symplectic potential verifies the canonical transformation explicitly:
\begin{align}
 P_\Phi\,\ed\Phi+P_a\,\ed a
 =\Pi_Y\,\ed T_Y+\Pi_a\,\ed a
   -\ed\!\left(\frac{a\Phi K}{g}\right).
 \label{eq:york-symplectic}
\end{align}
Thus \(\Pi_Y=a\Phi/g\) is conjugate to \(T_Y=K\), up to the displayed exact term and the simultaneous redefinition of the momentum of \(a\).

In these variables the gravitational lapse constraint is
\begin{align}
 \mathcal C_{\rm grav}
 =aT_Y\Pi_a-T_Y^2\Pi_Y
 +\frac1g\partial_\chi\!\left[
   \frac1a\partial_\chi\!\left(\frac{g\Pi_Y}{a}\right)\right]
 +\frac{\Pi_Y}{L^2}.
 \label{eq:york-unreduced-constraint}
\end{align}
It contains the scale momentum and spatial derivatives of \(\Pi_Y\). Deparametrization therefore requires gauge fixing, solution of the momentum constraint, boundary reduction, and a monotonicity analysis for \(T_Y\). On the symmetric slice \(T_Y=0\).

York-time quantizations of JT gravity exist in other settings \cite{ParrikarSake}, and a recent subregion construction derives a crossed-product constraint using extrinsic curvature as a physical clock \cite{BlommaertChen}. Adapting these constructions to global \(\mathrm{dS}_2\) requires its positive-curvature reduced phase space and regional boundary conditions.

\subsection{Implications for the Clock Construction}
\label{app:adm-scope}

The calculation gives the gravity-plus-matter constraints \eqref{eq:total-adm-constraints}, the candidate pair \((T_Y,\Pi_Y)\), the symmetric-slice balance \eqref{eq:boundary-term} and the classical relation \(Q^R_{\rm JT}=H^R_{\rm mat}\) of \Cref{prop:classical-charge-relation}. The auxiliary generator \(h\), the York momentum \(\Pi_Y\), the endpoint charge \(Q^R_{\rm JT}\), and the modular Hamiltonian \(K_R\) are distinct at this stage. Deriving a relation of the form \(h+K_u=0\) requires a reduced phase space, a monotone gravitational clock, quantization of the relational observables, and embeddings compatible with regional inclusion.

For JT gravity with an asymptotic boundary this programme has been carried out. Penington and Witten quantize the theory canonically and identify the algebra of boundary observables, which is generated by the ADM Hamiltonian in the pure theory and becomes a type-\(\mathrm{II}_\infty\) factor once bulk matter is coupled in \cite{PeningtonWitten2023}. Global \(\mathrm{dS}_2\) instead calls for a subregion assignment in a closed universe, together with the compatibility addressed by \Cref{prop:net-compatibility-fails}.

\bibliographystyle{unsrtnat}
\bibliography{refs}

@article{Maldacena,
  author        = {Maldacena, Juan M.},
  title         = {The Large {$N$} Limit of Superconformal Field Theories and
                   Supergravity},
  journal       = {Adv. Theor. Math. Phys.},
  volume        = {2},
  number        = {2},
  pages         = {231--252},
  year          = {1998},
  doi           = {10.4310/ATMP.1998.v2.n2.a1},
  eprint        = {hep-th/9711200},
  archivePrefix = {arXiv},
  primaryClass  = {hep-th},
  note          = {arXiv:hep-th/9711200}
}

@article{GKP,
  author        = {Gubser, S. S. and Klebanov, I. R. and Polyakov, A. M.},
  title         = {Gauge Theory Correlators from Noncritical String Theory},
  journal       = {Phys. Lett. B},
  volume        = {428},
  pages         = {105--114},
  year          = {1998},
  doi           = {10.1016/S0370-2693(98)00377-3},
  eprint        = {hep-th/9802109},
  archivePrefix = {arXiv},
  primaryClass  = {hep-th},
  note          = {arXiv:hep-th/9802109}
}

@article{Witten,
  author        = {Witten, Edward},
  title         = {Anti-{de Sitter} Space and Holography},
  journal       = {Adv. Theor. Math. Phys.},
  volume        = {2},
  number        = {2},
  pages         = {253--291},
  year          = {1998},
  doi           = {10.4310/ATMP.1998.v2.n2.a2},
  eprint        = {hep-th/9802150},
  archivePrefix = {arXiv},
  primaryClass  = {hep-th},
  note          = {arXiv:hep-th/9802150}
}

@article{RyuTakayanagi,
  author        = {Ryu, Shinsei and Takayanagi, Tadashi},
  title         = {Holographic Derivation of Entanglement Entropy from
                   {AdS/CFT}},
  journal       = {Phys. Rev. Lett.},
  volume        = {96},
  pages         = {181602},
  year          = {2006},
  doi           = {10.1103/PhysRevLett.96.181602},
  eprint        = {hep-th/0603001},
  archivePrefix = {arXiv},
  primaryClass  = {hep-th},
  note          = {arXiv:hep-th/0603001}
}

@article{HRT,
  author        = {Hubeny, Veronika E. and Rangamani, Mukund and
                   Takayanagi, Tadashi},
  title         = {A Covariant Holographic Entanglement Entropy Proposal},
  journal       = {JHEP},
  volume        = {07},
  pages         = {062},
  year          = {2007},
  doi           = {10.1088/1126-6708/2007/07/062},
  eprint        = {0705.0016},
  archivePrefix = {arXiv},
  primaryClass  = {hep-th},
  note          = {arXiv:0705.0016}
}

@article{EngelhardtWall,
  author        = {Engelhardt, Netta and Wall, Aron C.},
  title         = {Quantum Extremal Surfaces: Holographic Entanglement Entropy
                   Beyond the Classical Regime},
  journal       = {JHEP},
  volume        = {01},
  pages         = {073},
  year          = {2015},
  doi           = {10.1007/JHEP01(2015)073},
  eprint        = {1408.3203},
  archivePrefix = {arXiv},
  primaryClass  = {hep-th},
  note          = {arXiv:1408.3203}
}

@article{JLMS,
  author        = {Jafferis, Daniel L. and Lewkowycz, Aitor and
                   Maldacena, Juan and Suh, S. Josephine},
  title         = {Relative Entropy Equals Bulk Relative Entropy},
  journal       = {JHEP},
  volume        = {06},
  pages         = {004},
  year          = {2016},
  doi           = {10.1007/JHEP06(2016)004},
  eprint        = {1512.06431},
  archivePrefix = {arXiv},
  primaryClass  = {hep-th},
  note          = {arXiv:1512.06431}
}

@article{DHW,
  author        = {Dong, Xi and Harlow, Daniel and Wall, Aron C.},
  title         = {Reconstruction of Bulk Operators Within the Entanglement
                   Wedge in Gauge-Gravity Duality},
  journal       = {Phys. Rev. Lett.},
  volume        = {117},
  number        = {2},
  pages         = {021601},
  year          = {2016},
  doi           = {10.1103/PhysRevLett.117.021601},
  eprint        = {1601.05416},
  archivePrefix = {arXiv},
  primaryClass  = {hep-th},
  note          = {arXiv:1601.05416}
}

@article{ADH,
  author        = {Almheiri, Ahmed and Dong, Xi and Harlow, Daniel},
  title         = {Bulk Locality and Quantum Error Correction in {AdS/CFT}},
  journal       = {JHEP},
  volume        = {04},
  pages         = {163},
  year          = {2015},
  doi           = {10.1007/JHEP04(2015)163},
  eprint        = {1411.7041},
  archivePrefix = {arXiv},
  primaryClass  = {hep-th},
  note          = {arXiv:1411.7041}
}

@article{HarlowTASI,
  author        = {Harlow, Daniel},
  title         = {{TASI} Lectures on the Emergence of Bulk Physics in
                   {AdS/CFT}},
  journal       = {PoS},
  volume        = {TASI2017},
  pages         = {002},
  year          = {2018},
  doi           = {10.22323/1.305.0002},
  eprint        = {1802.01040},
  archivePrefix = {arXiv},
  primaryClass  = {hep-th},
  note          = {arXiv:1802.01040}
}

@article{DonnellyFreidel,
  author        = {Donnelly, William and Freidel, Laurent},
  title         = {Local Subsystems in Gauge Theory and Gravity},
  journal       = {JHEP},
  volume        = {09},
  pages         = {102},
  year          = {2016},
  doi           = {10.1007/JHEP09(2016)102},
  eprint        = {1601.04744},
  archivePrefix = {arXiv},
  primaryClass  = {hep-th},
  note          = {arXiv:1601.04744}
}

@article{BP,
  author        = {Bousso, Raphael and Penington, Geoff},
  title         = {Entanglement Wedges for Gravitating Regions},
  journal       = {Phys. Rev. D},
  volume        = {107},
  number        = {8},
  pages         = {086002},
  year          = {2023},
  doi           = {10.1103/PhysRevD.107.086002},
  eprint        = {2208.04993},
  archivePrefix = {arXiv},
  primaryClass  = {hep-th},
  note          = {arXiv:2208.04993}
}

@article{BK,
  author        = {Bousso, Raphael and Kaya, Sami},
  title         = {Fundamental Complement of a Gravitating Region},
  journal       = {Gen. Rel. Grav.},
  volume        = {57},
  number        = {8},
  pages         = {124},
  year          = {2025},
  doi           = {10.1007/s10714-025-03462-6},
  eprint        = {2505.15886},
  archivePrefix = {arXiv},
  primaryClass  = {hep-th},
  note          = {arXiv:2505.15886}
}

@article{HollowGrams,
  author        = {Kaya, Sami and Rath, Pratik and Ritchie, Kyle},
  title         = {{Hollow-grams}: Generalized Entanglement Wedges from the
                   Gravitational Path Integral},
  journal       = {JHEP},
  volume        = {09},
  pages         = {032},
  year          = {2025},
  doi           = {10.1007/JHEP09(2025)032},
  eprint        = {2506.10064},
  archivePrefix = {arXiv},
  primaryClass  = {hep-th},
  note          = {arXiv:2506.10064}
}

@misc{BalasubramanianCummings,
  author        = {Balasubramanian, Vijay and Cummings, Charlie},
  title         = {The Entropy of Finite Gravitating Regions},
  year          = {2023},
  eprint        = {2312.08434},
  archivePrefix = {arXiv},
  primaryClass  = {hep-th},
  howpublished  = {arXiv:2312.08434 [hep-th]},
  note          = {Preprint; no journal publication located}
}

@article{SahuGEWAlgebras,
  author        = {Sahu, Abhisek and {van der Heijden}, Jeremy and
                   {Van Raamsdonk}, Mark and Zibakhsh, Rana},
  title         = {Algebras for Generalized Entanglement Wedges},
  journal       = {JHEP},
  volume        = {07},
  pages         = {192},
  year          = {2026},
  doi           = {10.1007/JHEP07(2026)192},
  eprint        = {2511.21852},
  archivePrefix = {arXiv},
  primaryClass  = {hep-th},
  note          = {arXiv:2511.21852}
}

@unpublished{Y2H-Page_from_holograms,
  author        = {ElSayed, Hassan and Shiu, Gary},
  title         = {Page Curve of Black Hole Evaporation into a Gravitating
                   {de Sitter} Region from Holograms},
  year          = {2026},
  note          = {In preparation}
}

@article{CLPW,
  author        = {Chandrasekaran, Venkatesa and Longo, Roberto and
                   Penington, Geoff and Witten, Edward},
  title         = {An Algebra of Observables for {de Sitter} Space},
  journal       = {JHEP},
  volume        = {02},
  pages         = {082},
  year          = {2023},
  doi           = {10.1007/JHEP02(2023)082},
  eprint        = {2206.10780},
  archivePrefix = {arXiv},
  primaryClass  = {hep-th},
  note          = {arXiv:2206.10780}
}

@article{Witten2022,
  author        = {Witten, Edward},
  title         = {Gravity and the Crossed Product},
  journal       = {JHEP},
  volume        = {10},
  pages         = {008},
  year          = {2022},
  doi           = {10.1007/JHEP10(2022)008},
  eprint        = {2112.12828},
  archivePrefix = {arXiv},
  primaryClass  = {hep-th},
  note          = {arXiv:2112.12828}
}

@article{LeutheusserLiu2021,
  author        = {Leutheusser, Samuel and Liu, Hong},
  title         = {Causal Connectability Between Quantum Systems and the Black
                   Hole Interior in Holographic Duality},
  journal       = {Phys. Rev. D},
  volume        = {108},
  number        = {8},
  pages         = {086019},
  year          = {2023},
  doi           = {10.1103/PhysRevD.108.086019},
  eprint        = {2110.05497},
  archivePrefix = {arXiv},
  primaryClass  = {hep-th},
  note          = {arXiv:2110.05497}
}

@article{LeutheusserLiu2022,
  author        = {Leutheusser, Samuel and Liu, Hong},
  title         = {Emergent Times in Holographic Duality},
  journal       = {Phys. Rev. D},
  volume        = {108},
  number        = {8},
  pages         = {086020},
  year          = {2023},
  doi           = {10.1103/PhysRevD.108.086020},
  eprint        = {2112.12156},
  archivePrefix = {arXiv},
  primaryClass  = {hep-th},
  note          = {arXiv:2112.12156}
}

@misc{PeningtonWitten2023,
  author        = {Penington, Geoff and Witten, Edward},
  title         = {Algebras and States in {JT} Gravity},
  year          = {2023},
  eprint        = {2301.07257},
  archivePrefix = {arXiv},
  primaryClass  = {hep-th},
  howpublished  = {arXiv:2301.07257 [hep-th]},
  note          = {Preprint; no journal publication located}
}

@book{Takesaki,
  author        = {Takesaki, Masamichi},
  title         = {Theory of Operator Algebras {I}, {II}, {III}},
  publisher     = {Springer},
  address       = {Berlin},
  series        = {Encyclopaedia of Mathematical Sciences 124, 125, 127},
  year          = {2002--2003},
  note          = {Vol.~I first published 1979.  Crossed products and
                   Takesaki duality: Vol.~II, Ch.~X.  Structure of type-III
                   algebras and the continuous core: Vol.~II, Ch.~XII}
}

@article{Takesaki1972,
  author        = {Takesaki, Masamichi},
  title         = {Conditional Expectations in von {N}eumann Algebras},
  journal       = {J. Funct. Anal.},
  volume        = {9},
  number        = {3},
  pages         = {306--321},
  year          = {1972},
  doi           = {10.1016/0022-1236(72)90004-3}
}

@article{BrunettiGuidoLongo1993,
  author        = {Brunetti, Romeo and Guido, Daniele and Longo, Roberto},
  title         = {Modular Structure and Duality in Conformal Quantum Field
                   Theory},
  journal       = {Commun. Math. Phys.},
  volume        = {156},
  number        = {1},
  pages         = {201--219},
  year          = {1993},
  doi           = {10.1007/BF02096738},
  eprint        = {funct-an/9302008},
  archivePrefix = {arXiv},
  note          = {arXiv:funct-an/9302008}
}

@article{GabbianiFrohlich1993,
  author        = {Gabbiani, Fabrizio and Fr{\"o}hlich, J{\"u}rg},
  title         = {Operator Algebras and Conformal Field Theory},
  journal       = {Commun. Math. Phys.},
  volume        = {155},
  number        = {3},
  pages         = {569--640},
  year          = {1993},
  doi           = {10.1007/BF02096729}
}

@article{GuidoLongoWiesbrock1998,
  author        = {Guido, Daniele and Longo, Roberto and
                   Wiesbrock, Hans-Werner},
  title         = {Extensions of Conformal Nets and Superselection Structures},
  journal       = {Commun. Math. Phys.},
  volume        = {192},
  number        = {1},
  pages         = {217--244},
  year          = {1998},
  doi           = {10.1007/s002200050297},
  eprint        = {hep-th/9703129},
  archivePrefix = {arXiv},
  note          = {arXiv:hep-th/9703129}
}

@article{BuchholzDAntoniLongo2007,
  author        = {Buchholz, Detlev and D'Antoni, Claudio and
                   Longo, Roberto},
  title         = {Nuclearity and Thermal States in Conformal Field Theory},
  journal       = {Commun. Math. Phys.},
  volume        = {270},
  number        = {1},
  pages         = {267--293},
  year          = {2007},
  doi           = {10.1007/s00220-006-0127-9},
  eprint        = {math-ph/0603083},
  archivePrefix = {arXiv},
  note          = {arXiv:math-ph/0603083}
}

@article{KLM2001,
  author        = {Kawahigashi, Yasuyuki and Longo, Roberto and
                   M{\"u}ger, Michael},
  title         = {Multi-Interval Subfactors and Modularity of Representations
                   in Conformal Field Theory},
  journal       = {Commun. Math. Phys.},
  volume        = {219},
  number        = {3},
  pages         = {631--669},
  year          = {2001},
  doi           = {10.1007/PL00005565},
  eprint        = {math/9903104},
  archivePrefix = {arXiv},
  note          = {arXiv:math/9903104}
}

@unpublished{LongoLectureNotes,
  author        = {Longo, Roberto},
  title         = {Lectures on Conformal Nets, {P}art {II}: {N}ets of von
                   {N}eumann Algebras},
  year          = {2008},
  note          = {Unpublished draft, available at
                   \url{https://www.mat.uniroma2.it/longo/lecture-notes.html};
                   Part~I appeared as ``Real Hilbert subspaces, modular
                   theory, {$SL(2,\mathbb{R})$} and {CFT}'' in
                   \emph{Von Neumann Algebras in Sibiu}, Theta Ser. Adv.
                   Math.~10, Theta, Bucharest, pp.~33--91}
}

@article{Teitelboim,
  author        = {Teitelboim, Claudio},
  title         = {Gravitation and {H}amiltonian Structure in Two Space-Time
                   Dimensions},
  journal       = {Phys. Lett. B},
  volume        = {126},
  pages         = {41--45},
  year          = {1983},
  doi           = {10.1016/0370-2693(83)90012-6}
}

@article{Jackiw,
  author        = {Jackiw, Roman},
  title         = {Lower Dimensional Gravity},
  journal       = {Nucl. Phys. B},
  volume        = {252},
  pages         = {343--356},
  year          = {1985},
  doi           = {10.1016/0550-3213(85)90448-1}
}

@article{MTY,
  author        = {Maldacena, Juan and Turiaci, Gustavo J. and Yang, Zhenbin},
  title         = {Two Dimensional Nearly {de Sitter} Gravity},
  journal       = {JHEP},
  volume        = {01},
  pages         = {139},
  year          = {2021},
  doi           = {10.1007/JHEP01(2021)139},
  eprint        = {1904.01911},
  archivePrefix = {arXiv},
  primaryClass  = {hep-th},
  note          = {arXiv:1904.01911}
}

@article{CotlerJensenMaloney,
  author        = {Cotler, Jordan and Jensen, Kristan and Maloney, Alexander},
  title         = {Low-Dimensional {de Sitter} Quantum Gravity},
  journal       = {JHEP},
  volume        = {06},
  pages         = {048},
  year          = {2020},
  doi           = {10.1007/JHEP06(2020)048},
  eprint        = {1905.03780},
  archivePrefix = {arXiv},
  primaryClass  = {hep-th},
  note          = {arXiv:1905.03780}
}

@article{ParrikarSake,
  author        = {Parrikar, Onkar and Sake, Sunil Kumar},
  title         = {York Time in {JT} Gravity},
  journal       = {JHEP},
  volume        = {04},
  pages         = {010},
  year          = {2026},
  doi           = {10.1007/JHEP04(2026)010},
  eprint        = {2505.19231},
  archivePrefix = {arXiv},
  primaryClass  = {hep-th},
  note          = {arXiv:2505.19231}
}

@misc{BlommaertChen,
  author        = {Blommaert, Andreas and Chen, Chang-Han},
  title         = {Time in Gravitational Subregions and in Closed Universes},
  year          = {2026},
  eprint        = {2602.22153},
  archivePrefix = {arXiv},
  primaryClass  = {hep-th},
  howpublished  = {arXiv:2602.22153 [hep-th]},
  note          = {Preprint}
}

@article{NavarroSalasTalavera,
  author        = {Navarro-Salas, Jos{\'e} and Talavera, C. F.},
  title         = {Quantum Cosmological Approach to {2D} Dilaton Gravity},
  journal       = {Nucl. Phys. B},
  volume        = {423},
  number        = {2--3},
  pages         = {686--704},
  year          = {1994},
  doi           = {10.1016/0550-3213(94)90149-X},
  eprint        = {hep-th/9312030},
  archivePrefix = {arXiv},
  primaryClass  = {hep-th},
  note          = {arXiv:hep-th/9312030}
}

@article{KucharRomanoVaradarajan,
  author        = {Kucha{\v r}, Karel V. and Romano, Joseph D. and
                   Varadarajan, Madhavan},
  title         = {Dirac Constraint Quantization of a Dilatonic Model of
                   Gravitational Collapse},
  journal       = {Phys. Rev. D},
  volume        = {55},
  number        = {2},
  pages         = {795--808},
  year          = {1997},
  doi           = {10.1103/PhysRevD.55.795},
  eprint        = {gr-qc/9608011},
  archivePrefix = {arXiv},
  primaryClass  = {gr-qc},
  note          = {arXiv:gr-qc/9608011}
}

@article{BunchDavies,
  author        = {Bunch, T. S. and Davies, P. C. W.},
  title         = {Quantum Field Theory in {de Sitter} Space: Renormalization
                   by Point Splitting},
  journal       = {Proc. Roy. Soc. Lond. A},
  volume        = {360},
  number        = {1700},
  pages         = {117--134},
  year          = {1978},
  doi           = {10.1098/rspa.1978.0060}
}

@article{HolzheyWilczek,
  author        = {Holzhey, Christoph and Larsen, Finn and Wilczek, Frank},
  title         = {Geometric and Renormalized Entropy in Conformal Field
                   Theory},
  journal       = {Nucl. Phys. B},
  volume        = {424},
  number        = {3},
  pages         = {443--467},
  year          = {1994},
  doi           = {10.1016/0550-3213(94)90402-2},
  eprint        = {hep-th/9403108},
  archivePrefix = {arXiv},
  note          = {arXiv:hep-th/9403108}
}

@article{CalabreseCardy,
  author        = {Calabrese, Pasquale and Cardy, John L.},
  title         = {Entanglement Entropy and Quantum Field Theory},
  journal       = {J. Stat. Mech.},
  volume        = {0406},
  pages         = {P06002},
  year          = {2004},
  doi           = {10.1088/1742-5468/2004/06/P06002},
  eprint        = {hep-th/0405152},
  archivePrefix = {arXiv},
  note          = {arXiv:hep-th/0405152}
}

@article{AlvarezGaumeWitten1984,
  author        = {Alvarez-Gaum{\'e}, Luis and Witten, Edward},
  title         = {Gravitational Anomalies},
  journal       = {Nucl. Phys. B},
  volume        = {234},
  number        = {2},
  pages         = {269--330},
  year          = {1984},
  doi           = {10.1016/0550-3213(84)90066-X}
}

@article{Petz86,
  author        = {Petz, D{\'e}nes},
  title         = {Sufficient Subalgebras and the Relative Entropy of States
                   of a von {N}eumann Algebra},
  journal       = {Commun. Math. Phys.},
  volume        = {105},
  number        = {1},
  pages         = {123--131},
  year          = {1986},
  doi           = {10.1007/BF01212345}
}

@article{Petz88,
  author        = {Petz, D{\'e}nes},
  title         = {Sufficiency of Channels over von {N}eumann Algebras},
  journal       = {Quart. J. Math.},
  volume        = {39},
  number        = {1},
  pages         = {97--108},
  year          = {1988},
  doi           = {10.1093/qmath/39.1.97}
}

@book{WaldGR,
  author        = {Wald, Robert M.},
  title         = {General Relativity},
  publisher     = {University of Chicago Press},
  address       = {Chicago},
  year          = {1984},
  isbn          = {9780226870335}
}

@book{HawkingEllis,
  author        = {Hawking, S. W. and Ellis, G. F. R.},
  title         = {The Large Scale Structure of Space-Time},
  publisher     = {Cambridge University Press},
  address       = {Cambridge},
  year          = {1973},
  doi           = {10.1017/CBO9780511524646}
}

\end{document}